\documentclass[letterpaper, 11pt]{article}
\usepackage{amsmath,amssymb,graphicx,latexsym,amsthm,xspace}
\usepackage{framed}

\newtheorem{theorem}{Theorem}[section]
\newtheorem{lemma}[theorem]{Lemma}
\newtheorem{definition}[theorem]{Definition}

\newtheorem{corollary}[theorem]{Corollary}

\newcommand{\calA}{\mathcal{A}}
\newcommand{\calD}{\mathcal{D}}
\newcommand{\calI}{\mathcal{I}}
\newcommand{\calJ}{\mathcal{J}}
\newcommand{\calO}{\mathcal{O}}
\newcommand{\calP}{\mathcal{P}}

\newcommand{\calK}{\mathcal{K}}

\newcommand{\dtv}{d_{\mathrm{TV}}}

\newcommand{\bbP}{\mathbb{P}}
\newcommand{\bbR}{\mathbb{R}}
\newcommand{\bfX}{\mathbf{X}}
\newcommand{\biC}{\boldsymbol{C}}
\newcommand{\bib}{\boldsymbol{b}}
\newcommand{\bic}{\boldsymbol{c}}

\newcommand{\biw}{\boldsymbol{w}}
\newcommand{\bix}{\boldsymbol{x}}
\newcommand{\biy}{\boldsymbol{y}}
\newcommand{\biz}{\boldsymbol{z}}
\newcommand{\bimu}{\boldsymbol{\mu}}
\newcommand{\blp}{\textsf{BasicLP}\xspace}
\newcommand{\bsdp}{\textsf{BasicSDP}\xspace}

\newcommand{\lp}{\mathbf{lp}}
\newcommand{\opt}{\mathbf{opt}}
\newcommand{\val}{\mathbf{val}}

\newcommand{\olopt}{\overline{\mathbf{opt}}}
\newcommand{\ollp}{\overline{\mathbf{lp}}}
\newcommand{\olval}{\overline{\mathbf{val}}}

\newcommand{\poly}{\mathrm{poly}}
\newcommand{\E}{\mathop{\mathrm{E}}}

\title{Optimal Constant-Time Approximation Algorithms and (Unconditional) Inapproximability Results for Every Bounded-Degree CSP}
\author{Yuichi Yoshida\thanks{Supported by MSRA Fellowship 2010. This work was conducted while the author was visiting Rutgers University. }\\\\
  School of Informatics, Kyoto University, and\\ Preferred Infrastructure, Inc.\\yyoshida@lab2.kuis.kyoto-u.ac.jp}

\date{}
\begin{document}
\setcounter{page}{0}
\maketitle
\begin{abstract}
  Raghavendra (STOC 2008) gave an elegant and surprising result:
  if Khot's Unique Games Conjecture (STOC 2002) is true, 
  then for every constraint satisfaction problem (CSP), 
  the best approximation ratio is attained by a certain simple semidefinite programming and a rounding scheme for it.

  In this paper, 
  we show that similar results hold for constant-time approximation algorithms in the bounded-degree model.
  Specifically, we present the followings:
  (i) For every CSP, we construct an oracle that serves an access, 
  in constant time,
  to a nearly optimal solution to a basic LP relaxation of the CSP.
  (ii) Using the oracle, 
  we give a constant-time rounding scheme that achieves an approximation ratio coincident with the integrality gap of the basic LP.
  (iii) Finally, we give a generic conversion from integrality gaps of basic LPs to hardness results.
  All of those results are \textit{unconditional}.
  Therefore, for every bounded-degree CSP, 
  we give the best constant-time approximation algorithm among all.

  A CSP instance is called $\epsilon$-far from satisfiability if we must remove at least an $\epsilon$-fraction of constraints to make it satisfiable.
  A CSP is called testable if there is a constant-time algorithm that distinguishes satisfiable instances from $\epsilon$-far instances with probability at least $2/3$.
  Using the results above, 
  we also derive, under a technical assumption,
  an equivalent condition under which a CSP is testable in the bounded-degree model.
  \end{abstract}

\newpage

\section{Introduction}
In a \textit{constraint satisfaction problem} (CSP), 
the objective is to find an assignment to a set of variables that satisfies the maximum number of a given set of constraints on them.
Formally, a CSP~$\Lambda$ is specified by a set of predicates over alphabets $[q]=\{1,\ldots,q\}$.
Every instance of $\Lambda$ consists of a set of variables $V$,
and a set of constraints $\calP$ on them.
Each constraint consists of a predicate from $\Lambda$ applied to a subset of variables.
The objective is to find an assignment $\beta \in [q]^V$ to the variables that satisfies the maximum number of constraints. 
A large number of fundamental optimization problems,
such as \textsf{Max Cut} and \textsf{Max $k$-Sat},
are examples of CSPs.

Approximation algorithms for CSPs have been intensively studied.
Goemans and Williamson~\cite{GW95} first exploited semidefinite programmings (SDP) to \textsf{Max Cut} and \textsf{Max 2SAT} achieving the approximation ratio $\approx 0.878$.
After this breakthrough,
plethora of approximation algorithms using SDPs have been developed~\cite{KMS98,LLZ02}.
For inapproximability side,
tight hardness results have been successfully obtained for some important optimization problems such as \textsf{Max 3SAT}~\cite{Has01}.
However,
the approximability of many interesting CSPs such as \textsf{Max Cut} and \textsf{Max 2SAT} remains open.
Towards tightening this gap,
Khot~\cite{Kho02} introduced the Unique Games Conjecture (UGC).
Assuming the UGC,
tight hardness have been shown for \textsf{Max Cut}~\cite{KKMO04}, \textsf{Max 2SAT}~\cite{Aus07},
and \textsf{Max k-CSP}~\cite{AM08,ST06}.
Finally, 
Raghavendra~\cite{Rag08} succeeded to unify and generalize those approximation and inapproximability results for every CSP.
Specifically, Raghavendra showed that,
assuming the UGC,
for every CSP, 
a certain SDP combined with a certain rounding scheme attains the best approximation ratio among all polynomial-time approximation algorithms.
The ingenious technique in the proof is giving a generic conversion from integrality gaps of SDPs to hardness results via the UGC.

In this paper,
we are concerned with constant-time approximation algorithms CSPs.
That is, algorithms are supposed to run in time independent of sizes of instances.
We use the \textit{bounded-degree model},
which was originally introduced for graphs~\cite{GR08}.
In this model,
the number of alphabets, 
the maximum arity (the number of inputs to a predicate),
the maximum degree (the number of constraints where a variable appears),
and the maximum weight of constraints are bounded by constants.
Let $\calI$ be a $\Lambda$-CSP instance.
Since a constant-time algorithm cannot read the whole $\calI$,
we assume the existence of an oracle $\calO_{\calI}$ with which we can get information of $\calI$.
By specifying a variable $v$ and an index $i$, $\calO_{\calI}$ returns a constraint $P$ where $P$ is the $i$-th constraint where $v$ appears.
The efficiency of an algorithm is measured by the number of accesses to $\calO_{\calI}$, which is called \textit{query complexity}.

In this paper, we show an analogous result to Raghavendra's result:
for every CSP,
a certain linear programming (LP) combined with a certain rounding scheme attains the best approximation ratio among all constant-time approximation algorithms.
Furthermore, our results are \textit{unconditional}.
To give the statements precisely, 
we need to define several notions.
For a $\Lambda$-CSP instance $\calI$ with the variable set $V$ and the constraint set $\calP$,
there is a natural generic LP relaxation which we call \blp (see Section~\ref{sec:pre}).
Let $\lp(\calI)$ denote the objective value of an optimal solution to \blp for $\calI$,
$\opt(\calI)$ denote the value of an optimal solution of $\calI$,
and $\val(\calI,\beta)$ denote the value obtained by an assignment $\beta \in [q]^V$.
We define $\biw_{\calI}$ as the sum of weights of constraints in $\calI$.
Then, we define $\ollp(\calI)=\lp(\calI)/\biw_{\calI},\olopt(\calI)=\opt(\calI)/\biw_{\calI}$ and $\olval(\calI,\beta)=\val(\calI,\beta)/\biw_{\calI}$.
The \textit{integrality gap curve} $S_{\Lambda}(c)$ and the \textit{integrality gap} $\alpha_\Lambda$ of a CSP~$\Lambda$ is defined as 
\begin{eqnarray*}
  S_{\Lambda}(c)=\inf_{\calI\in \Lambda, \ollp(\calI) \geq c}\olopt(\calI),\quad \alpha_\Lambda = \inf_{\calI \in \Lambda} \olopt(\calI)/\ollp(\calI).
\end{eqnarray*}
The first result of this paper gives a tight approximation algorithm for every CSP.
\begin{theorem}\label{thr:upper}
  In the bounded-degree model,
  for every CSP~$\Lambda$ and $\epsilon>0$,
  there exists an algorithm that,
  given a $\Lambda$-CSP instance $\calI$ with $n$ variables and $\ollp(\calI)=c \in (0,1]$,
  with probability at least $2/3$,
  outputs a value $x$ such that $S_{\Lambda}(c-\epsilon)\biw_{\calI}-\epsilon n\leq x \leq \opt(\calI)$.
  Also, 
  for some fixed assignment $\beta$ such that $S_{\Lambda}(c-\epsilon)\biw_{\calI}-\epsilon n\leq \val(\calI,\beta) \leq \opt(\calI)$,
  given a variable $v$ in $\calI$,
  it computes $\beta_v$ in constant time.
\end{theorem}
The algorithm computes $\beta_v$ by rounding an LP solution to \blp for $\calI$.
Note that, for an instance $\calI$ with $\ollp(\calI)=c$,
$S_{\Lambda}(c)\biw_{\calI}$ is the best value we can hope for from the definition of $S_{\Lambda}(c)$.
Thus, in this sense, we will give a \textit{optimal rounding scheme} for \blp.

We mention that the additive error $\epsilon n$ cannot be removed.
To see this, 
suppose an instance $\calI$ consisting of $n$ variables and only one constraint.
Then, we have to see this constraint to approximate $\opt(\calI)$ if we do not allow the additive error.
However, it obviously takes $\Omega(n)$ queries.

For hardness side, we show the following.
\begin{theorem}\label{thr:lower}
  In the bounded-degree model,
  for every CSP~$\Lambda, c\in [0,1]$ and $\epsilon>0$,
  there exists a $\delta>0$ such that
  any algorithm that,
  given an instance $\calI$ with $n$ variables and $\olopt(\calI)=c \in [0,1]$,
  with probability at least $2/3$,
  outputs a value $x$ such that $(S_{\Lambda}(c)+\epsilon)\biw_{\calI} - \delta n \leq x \leq \opt(\calI)$ requires $\Omega(\sqrt{n})$ queries.
\end{theorem}
Note that, using the algorithm in Theorem~\ref{thr:upper}, given an instance $\calI$,
we can distinguish the case $\opt(\calI) \geq c\biw_{\calI}$ from the case $\opt(\calI) \leq S_{\Lambda}(c-\epsilon)\biw_{\calI}-\epsilon n$
(Technically, we need that $S_{\Lambda}(c)$ is non-decreasing, but this is obvious from the definition).
On the contrary, Theorem~\ref{thr:lower} asserts that we cannot distinguish the case $\opt(\calI) \geq c\biw_{\calI}$ from the case $\opt(\calI) \leq (S_{\Lambda}(c)+\epsilon)\biw_{\calI} - \delta n$.
Thus, the algorithm given in Theorem~\ref{thr:upper} is not just the best among constant-time approximation algorithm using \blp,
but the best among all constant-time approximation algorithms.

A value $x$ is called an \textit{$(\alpha,\beta)$-approximation} to a value $x^*$ if it satisfies $\alpha x^*-\beta \leq x \leq x^*$.
An algorithm is called an \textit{$(\alpha,\beta)$-approximation algorithm} for a CSP~$\Lambda$ if,
given a $\Lambda$-CSP instance $\calI$,
it computes an $(\alpha,\beta)$-approximation to $\opt(\calI)$ with probability at least $2/3$~\cite{NO08,PR07}.
The following is an immediate corollary achieved by Theorems~\ref{thr:upper} and~\ref{thr:lower}.
\begin{corollary}
  In the bounded-degree model, 
  for every CSP~$\Lambda$ and $\epsilon>0$,
  there exists a constant-time $(\alpha_\Lambda-\epsilon,\epsilon n)$-approximation algorithm for the CSP $\Lambda$.
  On the other hand, 
  for every CSP~$\Lambda$ and $\epsilon>0$, there exists a $\delta>0$ such that 
  any $(\alpha_\Lambda+\epsilon,\delta n)$-approximation algorithm for the CSP $\Lambda$ requires $\Omega(\sqrt{n})$ queries.
\end{corollary}

Theorem~\ref{thr:lower} has much implication to property testing.
A $\Lambda$-CSP instance $\calI$ is called \textit{satisfiable} if there is an assignment to variables that satisfies all the constraints.
Also, $\calI$ is called \textit{$\epsilon$-far from satisfiability} if we must remove at least $\epsilon twn$ constraints to make it satisfiable,
where $t,w,n$ is the maximum degree, the maximum weight, and the number of variables, respectively.
An algorithm is called a \textit{testing algorithm} for (the satisfiability of) a CSP~$\Lambda$ if,
given a $\Lambda$-CSP instance,
it accepts with probability at least $2/3$ if the instance is satisfiable,
and rejects with probability at least $2/3$ if the instance is $\epsilon$-far from satisfiability.
Unlike the hardness result given in~\cite{Rag08}, 
Theorem~\ref{thr:lower} holds also for $c=1$,
i.e., satisfiable instances.
Using this observation, we have the following theorem.
\begin{theorem}\label{thr:prop}
  In the bounded-degree model, the following holds for a CSP~$\Lambda$.
  If $S_{\Lambda}(1) < 1$, then any testing algorithm for the CSP~$\Lambda$ requires $\Omega(\sqrt{n})$ queries.
  If $S_{\Lambda}(1) = 1$ and $S_{\Lambda}(c)$ is continuous at $c=1$,
  then there exists a constant-time testing algorithm for the CSP~$\Lambda$.
\end{theorem}
We mention that Theorem~\ref{thr:prop} gives an ``if and only if'' condition of the testability of CSPs when their integrality gap curves are continuous at the point one
while we are not aware of any CSP for which the curve is not continuous at that point.

We give two direct applications of Theorem~\ref{thr:prop}.
An instance of \textsf{2-SAT} is a CNF formula where each constraint consists of at most two literals.
It is known that $S_{\textsf{Max 2-SAT}}(1)=1/2$, and it follows that we need $\Omega(\sqrt{n})$ queries to test \textsf{2-SAT}.
On the contrary, \textsf{2-SAT} is known to be testable with $\tilde{O}(\sqrt{n})$ queries~\cite{GR99}.
This fact implies that the lower bound in Theorem~\ref{thr:lower} is almost tight.
An instance of \textsf{Horn Sat} is a CNF formula where each constraint has at most one positive literal,
From~\cite{Zwi98}, it is easy to derive that $S_{\textsf{Max Horn SAT}}(1)=1$ and $S_{\textsf{Max Horn SAT}}(c)$ is continuous at $c=1$.
Thus, \textsf{Horn SAT} is testable in constant time.

\vspace{-10pt}
\paragraph{Related Work:}
Subsequent to Raghavendra's work~\cite{Rag08},
under the UGC,
certain SDPs and LPs are shown to be the best approximation algorithms for several classes of problems,
such as 
graph labeling problems (including \textsf{$k$-Way Cut}, \textsf{$0$-Extension}, and \textsf{Metric Labeling})~\cite{MTRS08},
kernel clustering problems~\cite{KN09},
ordering CSPs (including \textsf{Maximum Acyclic Subgraph})~\cite{GMR08}, 
and strict monotone CSPs (including \textsf{Minimum Vertex Cover})~\cite{KMTV09}.

There have been many studies on constant-time approximation algorithms in the bounded-degree model.
For algorithmic side,
mainly graph problems have been studied,
e.g., \textsf{Minimum Spanning Tree}~\cite{CRT01},
\textsf{Minimum Vertex Cover}~\cite{NO08,PR07,YYI09}, 
\textsf{Maximum Matching}~\cite{NO08,YYI09}, 
\textsf{Maximum Independent Set}~\cite{Alo10},
and \textsf{Minimum Dominating Set}~\cite{NO08,YYI09}.
For inapproximability results of graph problems,
\textsf{Minimum Dominating Set}~\cite{Alo10} and \textsf{Maximum Independent Set}~\cite{Alo10,Yos10} have been considered.
For CSPs,
it is known that, for every $\epsilon>0$, 
there exists $\delta>0$ such that any $(1/2+\epsilon,\delta n)$-approximation algorithm for \textsf{Max E2LIN2} and $(7/8+\epsilon,\delta n)$-approximation algorithm for \textsf{Max E3SAT} require linear number of queries~\cite{BOT02}.

We can compute the optimal value of a CSP instance within an additive error $O(\epsilon n^s)$ by sampling $\poly(1/\epsilon)$ variables and by solving the induced problem,
where $n$ is the number of variables and $s$ is the maximum arity~\cite{AdlVKK03,AE02}.
Thus, it is easy to approximate the solution of a dense instance in constant time.
Hence, we are concerned with the bounded-degree model in this paper.

\vspace{-10pt}
\paragraph{Proof Overview:}
%% The proof of Theorem~\ref{thr:upper} consists of two parts.
%% The first part (Section~\ref{sec:lp}) states that we can compute a (nearly) optimal solution of \blp in constant-time.
%% Since the size of the solution itself is non-constant, 
%% we do not explicitly create the whole solution.
%% Instead,
%% we construct an oracle that returns the value of a variable when we specify it.
%% The second part (Section~\ref{sec:round}) states that,
%% given an oracle access to a (nearly) optimal solution of \blp,
%% we can (implicitly) compute a solution $\beta$ such that $S_{\Lambda}(c-\epsilon)\biw_{\calI}-\epsilon n \leq \val(\calI,\beta) \leq \opt(\calI)$ in constant time.
%% Also, we show that can restore $\beta_v$ from $v$ for any variable $v$ in constant time.
%% Note that, for an instance $\calI$ satisfying $\lp(\calI)=c\biw_{\calI}$,
%% $\S_{\Lambda}(c)$ is the best possible value we can hope for.
%% Thus, the second part gives the \textit{optimal rounding scheme} for \blp.
We describe a proof sketch of Theorem~\ref{thr:upper}.
Let $\calO_\calI$ be the oracle access to a $\Lambda$-CSP instance $\calI$.
First, 
we construct an oracle access $\calO_{\lp}$ to a nearly optimal solution to $\blp$ for $\calI$,
Namely, 
if we specify a variable in \blp, 
$\calO_{\lp}$ outputs its value by accessing $\calO_{\calI}$ constant number of times.
To this end, we use a distributed algorithm for packing/covering LP given in~\cite{KMW06}.
In the distributed setting,
a linear programing is bound to a graph $G=(V,E)$.
Each primal variable $\bix_i$ and each dual variable $\biy_j$ is associated with a vertex $v_i^p\in V$ and $v_j^d\in V$, respectively. 
There are edges between primal and dual vertices wherever the respective variables occur in the corresponding inequality.
Thus, $(v_i^p, v_j^d)\in E$ if and only if $\bix_i$ occurs in the $j$-th inequality of the primal.
Let $G_{v,k}$ denote the graph induced by vertices whose distance from $v$ is at most $k$.
Then, a \textit{distributed algorithm in $k$ rounds} works in such a way that each vertex outputs a value of the corresponding variable based on $G_{v,k}$.
In~\cite{KMW06},
it is shown that if the matrix in the LP is ``sparse,''
then there is a distributed algorithm that computes a nearly optimal solution to the LP in $k$ rounds,
where $k$ is an integer determined by the sparsity of the LP.
Suppose that the degree of the graph is bounded by $\Delta$.
Then, given a variable,
we can compute the value of it by performing $\Delta^k$ queries to $\calO_{\calI}$ by simulating the process of the distributed algorithm.
With this method, we achieve $\calO_{\lp}$.
Though \blp is not a packing/covering LP,
after applying several number of transformations,
we get a packing LP that has essentially the same behavior under approximation.
Technically, we need to show that \blp is robust in the sense that
even if we violate each constraint by small amount,
the optimal value does not significantly increase.
We finally mention that, a predicate can return values in $[-1,1]$ in \cite{Rag08} while it can only return $0$ or $1$ in this paper.
This restriction comes from that we cannot transform \blp to a packing LP anymore if we allow negative values.

Next, we exhibit a solution to the original instance by rounding the LP solution given by $\calO_{\lp}$.
In~\cite{RS09}, Raghavendra and Steurer considered a certain SDP relaxation, which we call \bsdp,
and showed an optimal rounding scheme for it.
That is, it achieves an approximation ratio coincident with the integrality gap of \bsdp.
Our proof is based on their work.
First, from an instance $\calI$ and its LP solution,
we create another instance $\calI'$ by merging variables of $\calI$ that are close in the LP solution
so that the number of variables in $\calI'$ become constant.
Though we cannot explicitly construct the whole $\calI'$ since the number of constraints is not constant,
we can enumerate variables in $\calI'$.
Then, we perform brute force search on $\calI'$.
Specifically, we estimate the value obtained by each assignment to variables in $\calI'$ by accessing the oracle $\calO_{\lp}$.
Let $\beta'^*$ be the assignment for $\calI'$ that takes the maximum among them.
Note that $\beta'^*$ can be unfolded to an assignment $\beta^*$ for $\calI$.
Then, with high probability, we have $S_{\Lambda}(c-\epsilon)-\epsilon n \leq \val(\calI,\beta^*) \leq \opt(\calI)$.
Since, from a variable $v$ in $\calI$, we can get the corresponding variable in $\calI'$ in constant time,
we can compute $\beta_v^*$ in constant time.
The crucial fact used here is that the LP optimum does not change significantly after merging variables.

Now, we describe a proof sketch of Theorem~\ref{thr:lower}.
Let $\calI$ be a $\Lambda$-CSP instance such that $\lp(I)=c\biw_{\calI}$ while $\opt(I)$ is arbitrarily close to $S_{\Lambda}(c)\biw_{\calI}$.
Also, let $(\bix^*,\bimu^*)$ be the optimal LP solution to $\blp$ for $\calI$.
First, we create a distribution of instances $\calD^{\opt}$ by blowing up variables of $\blp$.
With high probability, 
an instance $\calJ$ generated by $\calD^{\opt}$ satisfies that $\olopt(\calJ) \leq \olopt(\calI)+\epsilon$ where $\epsilon$ is  an arbitrarily small constant.
Next, using the LP solution $(\bix^*,\bimu^*)$,
we create another distribution of instances $\calD^{\lp}$, 
which has the property that for all $\calJ$ generated by $\calD^{\lp}$, 
$\olopt(\calJ)\geq \ollp(\calI)$.
From Yao's minimax principle,
by showing that any deterministic algorithm that distinguishes $\calD^{\opt}$ from $\calD^{\lp}$ with high probability requires $\Omega(\sqrt{n})$ queries,
we have the desired result.

\vspace{-10pt}
\paragraph{Organization:}
In Section~\ref{sec:pre},
we give notations and basic technical tools used in this paper.
In Section~\ref{sec:lp},
we present an oracle access $\calO_{\lp}$ to a (nearly) optimal solution to \blp.
Section~\ref{sec:round} is devoted to describe how to round the LP solution optimally and to prove Theorem~\ref{thr:upper}.
We give proofs of Theorems~\ref{thr:lower} and~\ref{thr:prop} in Section~\ref{sec:lower} and Appendix~\ref{apx:prop}, respectively.

\section{Preliminaries}\label{sec:pre}
\subsection{Definitions}
For an integer $k$, $[k]$ denotes the set $\{1,\ldots,k\}$.
The \textit{arity} of a predicate $P:[q]^k\to \{0,1\}$ is the number of inputs to $P$, i.e., $k$ here.
The \textit{degree} of a variable is the number of constraints where the variable appears.
For a constraint $P$ in a CSP instance,
$V(P)$ denotes the set of variables in $P$.
Let $\beta$ be a vector or a set indexed by elements of a set $V$.
For a subset $S\subseteq V$,
we define $\beta_{|S}=\{\beta_v\}_{v\in S}$.

\begin{definition}
A \textit{bounded-degree constraint satisfaction problem} $\Lambda$ is specified by $\Lambda=([q],s,t,w,\bbP)$,
where $[q]$ is a finite domain,
$s$ is the maximum arity of predicates,
$t$ is the maximum degree of variables, 
$w$ is the maximum weight of predicates,
and $\bbP=\{P:[q]^k\to \{0,1\} \mid k \leq s\}$ is a set of predicates.
\end{definition}
\begin{definition}
An instance $\calI$ of a CSP~$\Lambda=([q],s,t,w,\bbP)$ is given by $\calI=(V,\calP,\biw)$,
where 
\begin{itemize}
\setlength{\itemsep}{0pt}
\item $V=\{v_1,\ldots,v_n\}$ is a set of variables taking values over $[q]$,
\item $\calP$ is a set of constraints,
  consisting of predicates $P\in \bbP$ applied to sequences $S$ of variables $V$ of size at most $s$.
  More precisely, 
  when a predicate $P$ is applied to a sequence $S=\{i_1,\ldots,i_k\}\subseteq [n]^k$,
  $P$ takes variables $V_{|S}=\{v_{i_1},\ldots,v_{i_k}\}$ as the input. 
\item $\biw$ is a set of weights $\{\biw_P\}_{P\in \calP}$ assigned to each constraint $P\in \calP$, 
  where $1\leq \biw_P\leq w$.
\end{itemize}
The objective is to find an assignment to variables $\beta \in [q]^V$ that maximizes the total weight of satisfied constraints,
i.e., \( \sum_{P\in \calP}\biw_P P(\beta) \).
\end{definition}
\begin{definition}[Bounded-degree Model]
  In \textit{the bounded-degree model},
  an algorithm is given a CSP~$\Lambda = ([q],s,t,w,\bbP)$ and the set of variables $V$ beforehand.
  A $\Lambda$-CSP instance $\calI = (V,\calP,\biw)$ is represented by an oracle $\calO_{\calI}$ such that $\calO_{\calI}$, on two numbers $v \in V, i\in [t]$,
  returns $P \in \calP$ where $P$ is the $i$-th constraint where $v$ appears.
  If no such constraint exists, it returns a special character $\bot$.
  The \textit{query complexity} of an algorithm is the number of accesses to $\calO_{\calI}$.
\end{definition}

%% Let $V$ and $\calP$ be the variable set and the constraint set of $\calI$, respectively.
%% Also, let $t$ be the maximum degree of $\calI$.
%% Then, by specifying a variable $v\in V$ and an index $i\in [t]$, 
%% $\calO_{\calI}$ returns 
%% If there is no such constraint, $\calO_{\calI}$ returns a special symbol.

In this paper, when there is no ambiguity,
symbols $q,s,t$ and $w$ are used to denote the parameters of a considered CSP.
Also, symbols $n,\calO_{\calI},\biw_{\calI}$ are used to denote the number of variables, the oracle access, and the total weight of an input instance $\calI$, respectively.

We consider an LP relaxation for a CSP~$\Lambda$ as follows, which we call \blp.
\begin{eqnarray*}
  \begin{array}{lll}
    \max & \sum\limits_{P\in \calP}\biw_P \sum\limits_{\beta \in [q]^{V(P)}}P(\beta)\bimu_{P,\beta}\\  
    \mbox{s.t.} & \sum\limits_{a\in [q]}\bix_{v,a}=1 & \forall v \in V\\    
    & \sum\limits_{\beta \in [q]^{V(P)}, \beta_v=a}\bimu_{P,\beta}=\bix_{v,a} & \forall P\in \calP, v \in V(P), a\in [q]\\
    & \bix_{v,a}\geq 0 & \forall v \in V, a\in [q]\\
    & \bimu_{P,\beta}\geq 0 & \forall P\in \calP, \beta\in [q]^{V(P)}.
  \end{array}
\end{eqnarray*}
Here, $\bix_{v}=\{\bix_{v,a}\}_{a\in [q]}$ (resp., $\bimu_P=\{\bimu_{P,\beta}\}_{\beta\in [q]^{V(P)}}$) can be seen as a distribution over assignments to a variable $v\in V$ (resp., a constraint $P\in \calP$), 
and we often identify them as distributions.
For an LP solution $(\bix,\bimu)$,
we define $\val(\calI,\bix,\bimu)$ as the value of the objective function of \blp obtained by $(\bix,\bimu)$.
We call an LP solution $(\bix,\bimu)$ \textit{$\epsilon$-infeasible} if it satisfies constraints of the form $\bix_{v,a}\geq 0$ and $\bimu_{P,\beta}\geq 0$ and violates other constraints by at most $\epsilon$.
We call a solution to an LP \textit{$(\alpha,\beta)$-approximate} if the objective value obtained by the solution is an $(\alpha,\beta)$-approximation to the optimal value of the LP.

\subsection{Basic Tools}

As a simple application of Hoeffding's inequality, we obtain the following.
\begin{lemma}\label{lmm:stat}
  Suppose that we have an oracle access to a function $f:[n]\to [0,w]$.
  That is, by specifying $x\in [n]$ as a query, 
  we can see the value of $f(x)$.
  Then, by querying $O(\frac{w^2}{\epsilon^2}\log \frac{1}{\delta})$ times, 
  with probability at least $1-\delta$, 
  we can compute a $(1,\epsilon n)$-approximation to $\sum_i f(i)$.
  \qed
\end{lemma}

Let $\calI$ be a $\Lambda$-CSP instance.
Not surprisingly,
we cannot compute the optimal solution $(\bix^*,\bimu^*)$ of \blp for $\calI$ in constant time.
Even worse, it is also hard to obtain a feasible solution in constant time.
Instead, we will compute a feasible (nearly) optimal solution $(\bix,\bimu)$ of an LP obtained by relaxing equality constraints.
Though this is an infeasible solution in the original LP,
The following lemma states that $\val(\calI,\bix,\bimu)$ is close to $\lp(\calI)$.
The proof, which needs Fourier analysis, is given in Appendix~\ref{apx:robust}.
\begin{lemma}[Robustness of \blp]\label{lmm:robust}
  Let $\calI$ be a $\Lambda$-CSP instance.
  Suppose that $(\bix,\bimu)$ is an $\epsilon$-infeasible LP solution for $\calI$ of value $c\biw_{\calI}$.
  Then, it holds that
  \begin{eqnarray*}
    \ollp(\calI)\geq c-\epsilon\cdot\poly(qs).
  \end{eqnarray*}
\end{lemma}

\section{A $(1-\epsilon,\epsilon n)$-approximation algorithm for \blp}\label{sec:lp}
In this section, we show the following theorem.
\begin{theorem}\label{thr:lp}
  In the bounded-degree model, 
  given a $\Lambda$-CSP instance $\calI$,
  for any $\epsilon>0$,
  we can construct an oracle $\calO_{\lp}$ that gives an access to an $\epsilon$-feasible $(1-\epsilon,\epsilon n)$-approximate solution to \blp for $\calI$.
  For each query, the number of queries performed to $\calO_{\calI}$ is at most $\exp(\exp(\poly(qstw/\epsilon)))$.
\end{theorem}
A packing LP is a problem of maximizing $\bib^T\biz$ subject to $A^T\biz \leq \bic$ and $\biz \geq 0$,
where $A\in \bbR_+^{m \times n}$ is a non-negative matrix and $\bib, \bic\in \bbR_+^{n}$ are non-negative vectors.
%% A packing LP is expressed as follows.
%% \begin{eqnarray}
%%   \begin{array}{ll}
%%     \max & \bib^T\biz \\ 
%%     \mbox{s.t.} & A^T\biz \leq \bic \\
%%     & \biz \geq 0,
%%   \end{array}
%% \end{eqnarray}
%% where $A\in \bbR_+^{m \times n}$ is a non-negative matrix and $\bib, \bic\in \bbR_+^{n}$ are non-negative vectors.
There is a constant-round distributed algorithm to compute a nearly optimal solution to the packing LP (see Appendix~\ref{sec:packing-lp} for a formal statement).
When a variable $\bix_{v,a}$ or $\bimu_{P,\beta}$ is specified as a query,
we locally simulate the distributed algorithm and output the value for it.
The only issue is that \blp is not a packing LP.
In this section, we transform \blp to a packing LP, 
and we will show that we can restore a good approximation to \blp from an approximation to the resulting packing LP.
First, we substitute $\bix_{v,a}$ by $1-\bix_{v,a}$ and relax each equality constraint by $\epsilon$.
Then, we obtain the following LP.
\begin{eqnarray}
  \begin{array}{lll}
    \max & \sum\limits_{P\in \calP}\biw_P \sum\limits_{\beta\in [q]^{V(P)}}P(\beta)\bimu_{P,\beta} \\
    \mbox{s.t.} & |\sum\limits_{a\in[q]}\bix_{v,a} - (q-1) | \leq \epsilon & \forall v \in V\\    
    & | \bix_{v,a}+\sum\limits_{\beta \in [q]^{V(P)}, \beta_v=a}\bimu_{P,\beta} - 1 | \leq \epsilon & \forall P\in \calP, v\in V(P), a\in [q]\\
    & \bix_{v,a}\geq 0 & \forall v\in V,a\in[q] \\
    & \bimu_{P,\beta}\geq 0 & \forall P\in \calP, \beta\in [q]^{V(P)}.
  \end{array}
  \label{lp:blp-relaxed}
\end{eqnarray}
\begin{lemma}
  Let $\calI$ be a $\Lambda$-CSP instance and $(\bix,\bimu)$ be an $\epsilon$-infeasible solution to LP~\eqref{lp:blp-relaxed} of value $c\biw_{\calI}$.
  Then, $\ollp(\calI)\geq c-\epsilon\cdot \poly(qs)$ holds.
\end{lemma}
\begin{proof}
  Clearly, $(1-\bix,\bimu)$ is an $2\epsilon$-infeasible solution to \blp of value $c\biw_{\calI}$.
  From Lemma~\ref{lmm:robust},
  the lemma holds.
\end{proof}

Next, 
to make the directions of the inequalities the same,
we introduce a complement variable for each variable, i.e., we define $\overline{\bix}_{v,a}=1-\bix_{v,a}$ and $\overline{\bimu}_{P,\beta}=1-\bimu_{P,\beta}$.
However, such equality constraints cannot be used in a packing LP.
Thus, we relax those equality constraints again.
That is, we introduce constraints of the form $\bix_{v,a}+\overline{\bix}_{v,a}\leq 1$ and $\bimu_{P,\beta}+\overline{\bimu}_{P,\beta}\leq 1$.
Instead, to discourage them to become much smaller than one, 
we add additional terms to the objective function.
By letting $\biC=(C,C,\ldots,C)$ where $C=q^{O(s)}\poly(tw)/\epsilon^2$ is a large constant, we get the following LP.
\begin{eqnarray}
  \begin{array}{lll}
    \max & \sum\limits_{P\in \calP}\biw_P \sum\limits_{\beta\in [q]^{V(P)}}P(\beta)\bimu_{P,\beta} + \biC^T(\bix+\overline{\bix})+ \biC^T(\bimu+\overline{\bimu}),\\
    \mbox{s.t.} & \sum\limits_{a\in[q]}\bix_{v,a} \leq q-1 + \epsilon & \forall v\in V\\    
    & \sum\limits_{a\in[q]}\overline{\bix}_{v,a} \leq 1+\epsilon & \forall v\in V\\    
    & \bix_{v,a}+\sum\limits_{\beta \in [q]^{V(P)}, \beta_v=a}\bimu_{P,\beta} \leq 1+\epsilon & \forall P\in \calP, v\in V(P), a\in [q]\\
    & \overline{\bix}_{v,a}+\sum\limits_{x \in [q]^{V(P)}, \beta_v=a}\overline{\bimu}_{P,\beta} \leq q^{|V(P)|-1}+\epsilon & \forall P\in \calP, v\in V(P), a\in [q]\\
    & \bix_{v,a}+\overline{\bix}_{v,a} \leq 1, \quad \bix_{v,a}\geq 0, \quad \overline{\bix}_{v,a}\geq 0 & \forall v \in V, a\in[q]\\
    & \bimu_{P,\beta}+\overline{\bimu}_{P,\beta} \leq 1, \quad \bimu_{P,\beta}\geq 0, \quad \overline{\bimu}_{P,\beta}\geq 0 & \forall P\in \calP, \beta\in [q]^{V(P)}.\\
  \end{array}
  \label{lp:blp-relaxed-again}
\end{eqnarray}

Fortunately, 
the optimal solutions to LP~\eqref{lp:blp-relaxed-again} and LP~\eqref{lp:blp-relaxed} are essentially the same.
\begin{lemma}[Theorem~7 of~\cite{FS97}, in a special form]\label{lmm:fs97}
  Let $\calI$ be a $\Lambda$-CSP instance and $(\bix^*,\overline{\bix}^*,\bimu^*,\overline{\bimu}^*)$ be the optimal solution to LP~\eqref{lp:blp-relaxed-again} with value $c\biw_{\calI}+CN$ where $N$ is the number of variables in LP~\eqref{lp:blp-relaxed-again}.
  Then, $\bix^*+\overline{\bix}^*=\mathbf{1}$ and $\bimu^*+\overline{\bimu}^*=\mathbf{1}$ hold.
  Also, $(\bix^*,\bimu^*)$ is the optimal solution to LP~(\ref{lp:blp-relaxed}) with value $c\biw_{\calI}$.
  \qed
\end{lemma}

Now, using the distributed algorithm given by~\cite{KMW06}, 
we have the following lemma.
The analysis of the query complexity is tedious and the proof is given in Appendix~\ref{sec:packing-lp}.
\begin{lemma}\label{lmm:packing-lp}
  In the bounded-degree model,
  given a $\Lambda$-CSP instance $\calI$,
  for any $\epsilon>0$,
  we can construct an oracle that serves an access to $(\bix,\overline{\bix},\bimu,\overline{\bimu})$, 
  which is a feasible $(1-\epsilon,0)$-approximate solution to LP~\eqref{lp:blp-relaxed-again}.
  For each query, the number of queries performed to $\calO_{\calI}$ is at most $\exp(\poly(qstw/\epsilon))$.
\end{lemma}

\begin{proof}[Proof of Theorem~\ref{thr:lp}]
  Let $(\bix,\overline{\bix},\bimu,\overline{\bimu})$ be a feasible $(1-\epsilon',0)$-approximate solution obtained by Lemma~\ref{lmm:packing-lp},
  where $\epsilon'$ is a constant determined later.
  For notational simplicity,
  we write the objective function as $\biw^T \bimu + \biC^T(\biz+\overline{\biz})$ where $\biz=(\bix,\bimu)$.
  Let $(\biz^*,\overline{\biz}^*)$ be the optimal solution to LP~\eqref{lp:blp-relaxed-again}.
  From Lemma~\ref{lmm:fs97}, $\biz^*+\overline{\biz}^*=\mathbf{1}$.
  Also, let $N\leq qn+q^s\cdot tn= (q+tq^s)n$ be the number of variables in LP~\eqref{lp:blp-relaxed-again}.
  Then, we have
  \begin{eqnarray*}
    \biw^T \bimu +  \biC^T(\biz+\overline{\biz}) \geq (1-\epsilon')(\biw^T\bimu^*+\biC^T(\biz^*+\overline{\biz}^*)) = (1-\epsilon')(\biw^T\bimu^*+C N).
  \end{eqnarray*}
  Thus,
  \begin{eqnarray*}
    & \biC^T(\biz+\overline{\biz}) &
    \geq
    (1-\epsilon')CN + (1-\epsilon')\biw^T\mu^* - \biw^T\mu
    \geq 
    (1-\epsilon')CN - \biw_{\calI}
    \geq 
    (1-\epsilon' - w/C)CN, \\
    &\biw^T\bimu&
    \geq
    (1-\epsilon')\biw^T\bimu^* + (1-\epsilon')(CN-\biC^T(\biz+\overline{\biz})) - \epsilon' \biC^T(\biz+\overline{\biz})
    \geq
    (1-\epsilon')\biw^T\bimu^* - \epsilon'C N. 
  \end{eqnarray*}
  In the former inequality, we used the fact that $\biw_{\calI} \leq wN$.
  In the latter inequality, we used the fact that $CN \geq \biC^T(\biz+\overline{\biz})$.

  From the former inequality, we have
  \begin{eqnarray*}
    \mathbf{1}^T(\mathbf{1}-\biz-\overline{\biz}) \leq (\epsilon'+w/C)N.
  \end{eqnarray*}
  Let $S$ be the set of variables $\biz_i$ ($= \bix_{v,a}$ or $\bimu_{P,\beta}$) such that $(1-\biz_i-\overline{\biz_i}) \geq \epsilon''$ where $\epsilon''$ is a constant determined later.
  From Markov's inequality,
  we have $|S|\leq (\epsilon'+w/C)/\epsilon''\cdot N$.
  Let $S_{\bix}= S\cap \{\bix_{v,a}\}_{v\in V,a\in [q]}$ and $S_{\bimu}=S\cap \{\bimu_{P,\beta}\}_{P\in \calP,\beta\in [q]^{V(P)}}$.
  The variables in $S_{\bix}$ and $S_{\bimu}$ are problematic since constraints in LP~\eqref{lp:blp-relaxed-again} involving them are far from being satisfied.
  Thus, in what follows,
  we modify these variables and obtain nearly feasible solution to LP~\eqref{lp:blp-relaxed-again}.

  First, we construct variables $\{\bix'_{v,a}\}_{v\in V, a\in [q]}$ by setting $\bix'_{v,a}=\bix_{v,a}$ if none of $\{\bix_{v,a'}\}_{a'\in[q]}$ is in $S_{\bix}$ and $\bix'_{v,a}=1/q$ if otherwise.
  Then, we construct variables $\{\bimu'_{P,\beta}\}_{P\in V(P), \beta\in [q]^{V(P)}}$ as follows.
  If none of $\{\bix_{v,a}\}_{v\in V(P), a\in[q]}$ was modified in the previous step,
  we set $\bimu'_{P,\beta}=\bimu_{P,\beta}$.
  If otherwise, 
  we set the values of $\{\bimu'_{P,\beta}\}_{\beta\in [q]^{V(P)}}$ in such a way that the distribution $\bimu'_P$ becomes consistent with the product distribution determined by $\{\bix'_{v}\}_{v\in P}$.
  Note that each modification to $\bix$ in the previous step involves at most $2tq^s$ modifications to $\bimu$.

  We calculate the decrease of the objective function.
  The decrease caused by the modification to $\bix_{v,a}$ is at most $\sum_{P\ni v}\biw_P\leq tw$,
  and the decrease caused by the modification to $\bimu_{P,\beta}$ is at most $\biw_P \leq w$.
  Thus, the total decrease is at most $tw|S|+2twq^s|S| \leq (\epsilon'+w/C)/\epsilon''\cdot(1+2q^s)tw N$.
  
  Note that for each unmodified variable $\biz_i$ ($= \bix_{v,a}$ or $\bimu_{P,\beta}$), 
  $\biz'_i+\overline{\biz}'_i \geq 1-\epsilon''$ holds.
  Thus, $(\bix',\overline{\bix}',\bimu',\overline{\bimu}')$ is an $\epsilon''$-infeasible solution with value at least 
  \begin{eqnarray*}
    \biw^T \bimu' 
    &\geq&
    (1-\epsilon')\biw^T\bimu^*- \epsilon' C N-(\epsilon'+w/C)/\epsilon''\cdot (1+2q^s)tw N\\
    &\geq &
    (1-\epsilon')\biw^T\bimu^* - (\epsilon' C + (\epsilon'+w/C)/\epsilon''\cdot (1+2q^s)tw ) (q+tq^s)n.
  \end{eqnarray*}
  Thus, $(\bix',\bimu')$ is an $\epsilon''$-infeasible $(1-\epsilon',(\epsilon' C + (\epsilon'+w/C)/\epsilon''\cdot (1+2q^s)tw ) (q+tq^s)n)$-approximate solution.
  By choosing $\epsilon'=\epsilon^3/(q^{O(s)}\poly(tw))$ and $\epsilon''=\epsilon$,
  we have an $\epsilon$-infeasible $(1-\epsilon,\epsilon n)$-approximate solution.

  We need to look at $q$ variables $\{\bix_{v,a}\}_{a\in [q]}$ to decide the value of $\bix'_{v,a}$,
  and we need to look at at most $qs$ variables $\{\bix_{v,a}\}_{v\in V(P), a\in [q]}$ to decide the value of $\bimu'_{P,\beta}$.
  Thus, the number of queries performed to $\calO_{\calI}$ is at most 
  \(
  \max(q,qs)\exp(\poly(qstw/\epsilon'))=\exp(\exp(\poly(qstw/\epsilon)))
  \).
\end{proof}

\section{Optimal Rounding of \blp}\label{sec:round}
In this section, using $\calO_{\lp}$, we give an algorithm described in Theorem~\ref{thr:upper}.
Let $\calI=(V,\calP,\biw)$ be a $\Lambda$-CSP instance.
For a mapping $\phi:V\to V'$, 
we define a new $\Lambda$-CSP instance $\calI/\phi=(V',\calP',\biw')$ on the variable set $V'$ by identifying variables of $\calI$ that get mapped to the same variable in $V'$.
For each constraint $P\in \calP$ on the variable set $\{v_1,\ldots,v_k\}$ with weight $\biw_P$,
we have a constraint $P'\in \calP'$ on the variable set $\{\phi(v_1),\ldots,\phi(v_k)\}$ with weight $\biw_P$.
For $x\in [0,1]$, 
we define $x^\epsilon=(k+1)\epsilon$ where $k$ is the positive integer such that $k\epsilon < x \leq (k+1)\epsilon$.
We define $x^\epsilon=0$ when $x=0$.
In what follows, we assume that $1/\epsilon$ is an integer.
If not, we slightly decrease $\epsilon$ until $1/\epsilon$ become an integer.
Let $(\bix,\bimu)$ be an LP solution for $\calI$.
We identify variables $v$ of $\calI$ that have the same values $\{\bix_{v,a}^\epsilon\}_{a\in [q]}$.
Formally, we consider another $\Lambda$-CSP instance $\calI/\phi_{\bix}$ 
where $\phi_{\bix}:V\to \{0,\ldots,1/\epsilon\}^q$ is defined as \( \phi_{\bix}(v)=(\bix_{v,1}^\epsilon,\ldots,\bix_{v,q}^\epsilon) \).
We have following two lemmas, the proofs of which are in Appendix~\ref{apx:round-appendix}.
\begin{lemma}\label{lmm:discretize}
  Let $\calI$ be a $\Lambda$-CSP instance and $(\bix,\bimu)$ be an $\epsilon$-infeasible LP solution for $\calI$.
  Then, $(\bix^\epsilon,\bimu)$ is a $(q+1)\epsilon$-infeasible LP solution for $\calI$.
\end{lemma}
\begin{lemma}\label{lmm:compression}
  Let $\calI$ be a $\Lambda$-CSP instance and $(\bix,\bimu)$ be an $\epsilon$-infeasible $(1-\epsilon,\epsilon)$-approximate LP solution for $\calI$,
  where $\epsilon>0$ is a small constant.
  Then, the variable folding $\calI/\phi_{\bix}$ satisfies that
  \begin{itemize}
    \setlength{\itemsep}{0pt}
  \item $\lp(\calI/\phi_{\bix})\geq \lp(\calI)- \epsilon\cdot \poly(qstw) n$,
  \item The variable set of $\calI/\phi_{\bix}$ has a cardinality $\exp(\poly(q/\epsilon))$.
  \end{itemize}
\end{lemma}

\begin{proof}[Proof of Theorem~\ref{thr:upper}]
  Let $\epsilon'$ be a constant determined later and $(\bix,\bimu)$ be an $\epsilon'$-infeasible $(1-\epsilon',\epsilon')$-approximate solution for $\calI$.
  Consider a folded instance $\calI'=\calI/\phi_{\bix}$ on the variable set $V' := \{0,\ldots,1/\epsilon\}^q$.
  Since there are at most $\exp(\poly(q/\epsilon'))$ variables in $V'$,
  there are at most $N:=\exp(\exp(\poly(q/\epsilon')))$ assignments to $V'$.
  For each assignment $\beta' \in [q]^{V'}$, 
  we estimate the value $\val(\calI',\beta')$ as follows.
  First, we note that $\beta'$ can be unfolded to an assignment $\beta \in [q]^V$ to $\calI$ with the same value.
  Then, for each variable $v\in V$,
  we associate a value $f_v=\sum_{P\ni v}P(\beta)/|P|$.
  It is clear that $0\leq f_v \leq tw$ and $\sum_{v\in V}f_v=\val(\calI,\beta)=\val(\calI',\beta')$.
  Also, we can calculate the value $f_v$ by querying $\calO_{\lp}$ at most $qst$ times.
  Thus, using the algorithm given in Lemma~\ref{lmm:stat},
  we get a $(1,\epsilon n/2)$-approximation to $\val(\calI',\beta')$ with probability at least $1-1/3N$ by querying $\calO_{\lp}$ at most $\poly(qstw/\epsilon)O(\log N)$ times.

  By the union bound, 
  with probability at least $2/3$,
  we obtain a $(1,\epsilon n/2)$-approximation to $\val(\calI',\beta')$ for every assignment $\beta'$.
  Let $\beta'^* \in [q]^{V'}$ be the assignment that takes the maximum value among those assignments.
  Then, $\val(\calI',\beta'^*)$ is a $(1,\epsilon n/2)$-approximation to $\opt(\calI')$.
  The number of queries performed to $\calO_{\lp}$ is at most $\poly(qstw/\epsilon)O(N\log N)$.

  Let $\beta^* \in [q]^V$ be the unfolded assignment of $\beta'^*$.
  We can safely assume that $\biw_{\calI'} = \biw_{\calI} \geq \epsilon n$.
  If not, $\val(\calI,\beta^*)$ is indeed a $(1,\epsilon n)$-approximation to $\opt(\calI)$.
  When $\biw_{\calI'} = \biw_{\calI} \geq \epsilon n$, it holds that
  \begin{eqnarray*}
    \val(\calI,\beta^*) 
    &\geq &
    \opt(\calI')-\frac{\epsilon n }{2}
    \geq
    S_\Lambda(\ollp(\calI'))\biw_{\calI} - \frac{\epsilon n}{2}\\
    &\geq&
    S_{\Lambda}\left(\ollp(\calI)-\frac{\epsilon' \cdot \poly(qstw)n}{\biw_{\calI}}\right)\biw_{\calI} - \frac{\epsilon n}{2} \quad \text{(using Lemma~\ref{lmm:compression})}\\
    &\geq&
    S_{\Lambda}\left(\ollp(\calI)-\frac{\epsilon'\cdot \poly(qstw) }{\epsilon}\right)\biw_{\calI} - \frac{\epsilon n}{2} \quad \text{(using $\biw_{\calI} \geq \epsilon n$)}
  \end{eqnarray*}
  We are done by setting $\epsilon'=\epsilon^2/\poly(qstw)$.
  The number of queries performed to $\calO_{\calI}$ is at most
  \(
  \poly(qstw/\epsilon)O(N \log N)\cdot \exp(\exp(\poly(qstw/\epsilon'))) = \exp(\exp(\poly(qstw/\epsilon))) 
  \).
  Once we have fixed $\beta^*$, given a variable $v \in V$, we can compute $\beta^*_v$ by accessing $\calO_{\lp}$ $q$ times.
  The query complexity is at most $\exp(\exp(\poly(qstw/\epsilon)))$.
\end{proof}

\section{Lower Bounds}\label{sec:lower}
In this section, we prove Theorem~\ref{thr:lower}.
As we described in the introduction, we utilize Yao's minimax principle.
That is, we construct two distributions of instances such that they have much different optimal values and also it is hard to distinguish them in constant time.
We fix a $\Lambda$-CSP instance $\calI=(V,\calP,\biw)$ with the optimal LP solution $(\bix^*,\bimu^*)$ throughout this section.
To convert the LP integrality gap $\overline{\opt}(\calI)/\overline{\lp}(\calI)$ of $\calI$ to hardness results,
we construct two distributions $\calD_{N,T}^{\opt}$ and $\calD_{N,T}^{\lp}$ using $\calI$ and $(\bix^*,\bimu^*)$.
Here, $N$ and $T$ will determine the number of variables and the maximum degree of instances generated by $\calD_{N,T}^{\opt}$ and $\calD_{N,T}^{\lp}$, respectively.
We show that, by taking $T$ as a large constant (independent of $N$),
almost all instances $\calJ$ in $\calD_{N,T}^{\opt}$ satisfy that $\overline{\opt}(\calJ) \leq \overline{\opt}(\calI)+\epsilon$.
Also, we show that all instances $\calJ$ in $\calD_{N,T}^{\lp}$ satisfy that $\overline{\opt}(\calJ) \geq \overline{\lp}(\calI)$.
Finally, we define $\calD_{N,T}^\star$ as the distribution that chooses $\calD_{N,T}^{\opt}$ or $\calD_{N,T}^{\lp}$ randomly and outputs an instance generated by the chosen distribution.
Then, given an oracle access $\calO_{\calJ}$ to an instance $\calJ$ generated by $\calD_{N,T}^\star$,
a deterministic algorithm is supposed to guess the original distribution ($\calD_{N,T}^{\opt}$ or $\calD_{N,T}^{\lp}$) of $\calJ$ with probability at least $2/3$.
By showing that such an algorithm requires $\Omega(\sqrt{N})$ queries,
we conclude that any randomized algorithm that,
given an instance $\calJ$,
distinguishes the case $\olopt(\calJ) \geq \ollp(\calI)$ from the case $\olopt(\calJ) \leq \olopt(\calI)+\epsilon$
requires $\Omega(\sqrt{N})$ queries.
By choosing as $\calI$ an instance with the worst integrality gap,
we have the desired result.
\begin{figure}[t]
  \begin{center}
    \includegraphics[scale=0.5]{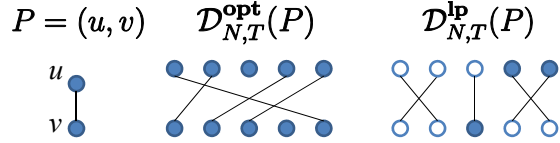}
    \caption{Construction of $\calD_{N,T}^{\opt}(P)$ and $\calD_{N,T}^{\lp}(P)$.
      Here, the alphabet size $q=2$, 
      and we choose $N=5$ and $T=1$.
      Also, $\bimu^{*}_{P,00}=0.4$, $\bimu^{*}_{P,01}=0.2$, $\bimu^{*}_{P,10}=0.4$, and $\bimu^*_{P,11}=0$.
      It follows that $\bix^{*}_{u,0}=0.6$, $\bix^{*}_{u,1}=0.4$, $\bix^{*}_{v,0}=0.8$, and $\bix^*_{v,1}=0.2$.
      White (resp., black) variables in $\calD_{N,T}^{\lp}(P)$ indicate that they are assigned to $0$ (resp., $1$).
    }
    \label{fig:lower}
  \end{center}
\end{figure}
\paragraph{Construction of $\calD_{N,T}^{\opt}$:}
Before stating the construction of $\calD_{N,T}^{\opt}$,
we introduce a distribution $\calD_{N,T}^{\opt}(P)$ for a constraint $P\in \calP$ applied to a sequence of variables $\{v_1,\ldots,v_k\}$ (see Fig.~\ref{fig:lower}). 
An instance $\calJ_P$ of $\calD_{N,T}^{\opt}(P)$ is generated as follows.
The variable set of $\calJ_P$ is $\{v_1,\ldots,v_k\} \times [N]$.
We naturally regard that the set of variables $V_i = \{(v_i,j) \mid j\in [N]\}$ corresponds to a variable $v_i$.
Next, we create $TN$ constraints among those variables.
To this end, after splitting each variable of $\calJ_P$ into $T$ copies,
we take random perfect $k$-partite matching in such a way that each matching takes one variable from each $V_i$.
For each such matching $\{u_1,\ldots,u_k\}$ where $u_i$ is of the form $(v_i,j_i)$,
we create a constraint $P(u_1,\ldots,u_k)$ of weight $\biw_P$.
Finally, we merge the split variables.

We define the distribution $\calD_{N,T}^{\opt}$ using $\calD_{N,T}^{\opt}(P)$.
An instance $\calJ$ of $\calD_{N,T}^{\opt}$ is generated as follows.
For each $P\in \calP$,
we create an instance $\calJ_P$ according to the distribution $\calD_{N,T}^{\opt}(P)$.
Then, $\calJ$ is a union of $\{\calJ_P\}_{P\in \calP}$ obtained by merging variable sets as follows.
Let $P_1,\ldots,P_\ell \in \calP$ be the set of constraints containing a variable $v\in V$.
We let $V_i (i \in [\ell])$ denote the set of variables in $\calJ_{P_i}$ corresponding to $v$.
Then, we take random perfect $\ell$-partite matching among $V_1,\ldots,V_\ell$,
and we merge variables in each matching.
We repeat the same process for every $v\in V$.
We note that the variable set of $\calJ$ is $V\times [N]$,
and the number of constraints in $\calJ$ is $|\calP|TN$.
Now, we decide the indices of constraints, 
which are used as arguments of the oracle access $\calO_{\calJ}$.
We use the following rule.
Suppose that $P$ is the $i$-th constraint where $v\in P$ appears (in the sense of $\calI$),
then for a variable $(v,j) (j \in [N])$, 
we randomly assign $T$ indices $\{(T-1)i+1,\ldots,Ti\}$ to designate $T$ constraints made by $\calJ_P$.
Finally, labels of vertices are randomly permuted.

\paragraph{Construction of $\calD_{N,T}^{\lp}$:}
Before stating the construction of $\calD_{N,T}^{\lp}$,
again we introduce another distribution $\calD_{N,T}^{\lp}(P)$ for a constraint $P\in \calP$ applied to a sequence of variables $\{v_1,\ldots,v_k\}$. (see Fig.~\ref{fig:lower}).
An instance $\calJ_P$ is generated as follows.
The variable set of $\calJ_P$ is $\{v_1,\ldots,v_k\}\times [N]$.
We naturally regard that the set of variables $V_i = \{(v_i,j) \mid j\in [N]\}$ corresponds to a variable $v_i$.
For each $\beta\in [q]^{V(P)}$,
we take a $\bimu^*_{P,\beta}$-fraction of variables from each $V_i$, and let $V_{i,\beta}\subseteq V_i$ denote the set of such variables.
Variables in $V_{i,\beta}$ are said to be \textit{assigned to $\beta_{v_i} \in [q]$}.
A subtlety here is that $\bimu_{P,\beta}N$ may not be an integer.
We ignore this issue for simplicity since we can make the error arbitrarily small by choosing $N$ large enough.
Next, we create $\bimu_{P,\beta}TN$ constraints among $V_{1,\beta},\ldots,V_{k,\beta}$.
To this end, after splitting each variable into $T$ copies,
we take random perfect $k$-partite matching in such a way that each matching takes one variable from each $V_{i,\beta}$.
For each matching $\{u_1,\ldots,u_k\}$ where $u_i$ is of the form $(v_i,j_i)$,
we create a constraint $P(u_1,\ldots,u_k)$ of weight $\biw_P$.
Finally, we merge the split variables again.
We note that, for any variable $v_i \in V(P)$,
an $\bix^*_{v_i,a}$-fraction of variables of $V_i$ is assigned to $a$.

We define the distribution $\calD_{N,T}^{\lp}$ using $\calD_{N,T}^{\lp}(P)$.
An instance $\calJ$ of $\calD_{N,T}^{\lp}$ is generated as follows.
For each $P\in \calP$,
we create an instance $\calJ_P$ according to the distribution $\calD_{N,T}^{\lp}(P)$.
Then, $\calJ$ is a union of $\{\calJ_P\}_{P\in \calP}$ obtained by merging variable sets as follows.
Let $P_1,\ldots,P_\ell \in \calP$ be the set of constraints containing a variable $v\in V$.
We let $V_{i,a} (i\in [\ell], a\in [q])$ denote the set of variables in $\calJ_{P_i}$ that correspond to $v$ and are assigned to $a$.
Note that the sizes of $V_{i,a} (i\in [\ell])$ are the same from the construction of $\calD_{N,T}^{\lp}(P_i)$.
We take random perfect $\ell$-partite matching among $V_{1,a},\ldots,V_{\ell,a}$ and we merge vertices in each matching.
We repeat the same process for every $v\in V$ and $a\in [q]$.
Note that the variable set of $\calJ$ is $V\times [N]$ and the number of constraints in $\calJ$ is $|\calP|TN$.
To decide the indices of constraints and labels of vertices, we use the same rule as $\calD_{N,T}^{\opt}$.

We have the following three lemmas, 
the proofs of which are in Appendix~\ref{apx:lower-appendix}.
\begin{lemma}\label{lmm:less-than-opt}
  For every $\epsilon>0$, there is a $T>0$ satisfying the following.
  Let $\calJ$ be a $\Lambda$-CSP instance generated by $\calD_{N,T}^{\opt}$.
  With probability $1-o(1)$,
  $\olopt(\calJ) \leq \olopt(\calI)+\epsilon$.
\end{lemma}
\begin{lemma}\label{lmm:more-than-lp}
  Let $\calJ$ be a $\Lambda$-CSP instance generated by $\calD_{N,T}^{\lp}$.
  Then, $\olopt(\calJ) \geq \ollp(\calI)$ holds.
\end{lemma}
\begin{lemma}\label{lmm:distinguish}
  In the bounded-degree model,
  any deterministic algorithm that,
  given an oracle access to $\calO_{\calJ}$ generated by $\calD_{N,T}^\star$,
  correctly guesses the original distribution of $\calJ$ with probability at least $3/5$ requires at least $\Omega(\sqrt{N})$ queries.
\end{lemma}

\begin{proof}[Proof of Theorem~\ref{thr:lower}]
  Let us fix $c\in [0,1]$ and $s=S_{\Lambda}(c)$.
  Then, there exists a $\Lambda$-CSP instance $\calI$ such that $\ollp(\calI)=c$ and $\olopt(\calI)$ is arbitrarily close to $s$.
  Suppose that there exists a deterministic algorithm $\calA$ with query complexity $o(\sqrt{n})$ that, 
  given an instance $\calJ$ of $n$ variables,
  with probability at least $2/3$,
  distinguishes the case $\opt(\calJ) \geq c\biw_{\calJ}$ from the case $\opt(\calJ) \leq (s+\epsilon)\biw_{\calJ}-\epsilon n$.
  Let $T$ be a constant given by Lemma~\ref{lmm:less-than-opt} by replacing $\epsilon$ with $\epsilon/2$.

  Suppose that $\calJ$ is generated by $\calD_{N,T}^{\opt}$.
  Then, from Lemma~\ref{lmm:less-than-opt}, 
  with probability at least $1-o(1)$, 
  it holds that $\opt(\calJ) \leq (s+\epsilon/2)\biw_{\calJ} = (s+\epsilon)\biw_{\calJ}-\epsilon \biw_{\calJ}/2 \leq (s+\epsilon)\biw_{\calJ} - \epsilon n$.
  In the last inequality, we use the fact that $\biw_{\calJ}/2 \geq \biw_{\calJ}/T \geq n$ when $\epsilon$ is small.
  Suppose that $\calJ$ is generated by $\calD_{N,T}^{\lp}$.
  Then, from Lemma~\ref{lmm:more-than-lp},
  it holds that $\opt(\calJ) \geq c\biw_{\calJ}$.

  Thus, in total, the algorithm outputs the correct answer with probability at least
  \( 1/2 \cdot (1-o(1)) \cdot 2/3 + 1/2\cdot 2/3 = 2/3-o(1) \).
  This contradicts Lemma~\ref{lmm:distinguish}.
\end{proof}

\newpage

\appendix
\noindent {\bf\Large \appendixname}

\section{Robustness of \blp}\label{apx:robust}
In this section, we give a proof of Lemma~\ref{lmm:robust}.
Our strategy is transforming $(\bix,\bimu)$ to a feasible solution without decreasing the LP value much.
In the first step,
we construct $\bix'$ from $\bix$ that satisfies $\sum_{a\in [q]}\bix'_{v,a}=1$ for every $v\in V$.
\begin{lemma}\label{lmm:surgery}
  Let $(\bix,\bimu)$ be an $\epsilon$-infeasible LP solution for a $\Lambda$-CSP instance $\calI$ where $\epsilon>0$ is a small constant.
  Then, $\bix$ can be transformed to $\bix'$ so that
  \begin{eqnarray}
    \sum_{a\in[q]}\bix'_{v,a}&=&1 \quad \forall v\in V, \label{eq:surgery-1}\\
    |\bix'_{v,a}-\bix_{v,a}| &\leq& 2\epsilon \quad \forall v\in V, a\in [q]. \label{eq:surgery-2}
  \end{eqnarray}
  In particular, 
  $(\bix',\bimu)$ is a $3\epsilon$-infeasible LP solution that satisfies $\sum_{a\in [q]}\bix'_{v,a}=1$ for every $v\in V$.
\end{lemma}
\begin{proof}
  We define $\bix'_{v,a}=\bix_{v,a}/\sum_{a\in [q]}\bix_{v,a}$.
  The condition (\ref{eq:surgery-1}) clearly holds.
  From the $\epsilon$-infeasibility of $\bix$, 
  $|\sum_{a\in [q]} \bix_{v,a}-1|\leq \epsilon $ holds.
  It follows that $|\bix'_{v,a}-\bix_{v,a}|\leq \epsilon/(1-\epsilon) \leq 2\epsilon$ when $\epsilon$ is small.
\end{proof}

In the second step,
we construct $\bimu'$ that satisfies $\sum_{\beta\in [q]^{V(P)},\beta_v=a}\bimu'_{P,\beta}=\bix'_{v,a}$ for all $P\in \calP, v\in V(P)$.
\begin{lemma}\label{lmm:smoothing}
  Let $(\bix,\bimu)$ be an $\epsilon$-infeasible solution for a $\Lambda$-CSP instance $\calI$ satisfying $\sum_{a\in [q]}\bix_{v,a}=1$ for every $v\in V$.
  Then, $\bimu$ can be transformed to $\bimu'$ so that
  \begin{eqnarray*}
    \Pr_{\beta\sim \bimu'_P}[\beta_v=a]&=&(1-\delta)\bix_{v,a}+\frac{\delta}{q} \quad \forall P\in \calP, v\in V(P), a\in [q], \\
    ||\bimu_P-\bimu'_P||_1&\leq& 2\delta \quad \forall P\in \calP.
  \end{eqnarray*}
  where $\delta=kq^3\epsilon$.
\end{lemma}
\begin{proof}
  Let us fix a predicate $P\in \calP$ and $S=V(P)$.
  We may assume $S=\{1,\ldots,k\}$ where $k\leq s$.
  We can think of $\bimu_P$ as a function $f:[q]^k\to \bbR$ such that $f(\beta)$ is the probability of the assignment $\beta$ under the distribution $\bimu_P$.

  Let $\chi_1,\ldots,\chi_q$ be an orthonormal basis of the vector space $\{f:[q]\to \bbR\}$ such that $\chi_1\equiv 1$.
  Here, orthonormal means that $E_{a\in[q]}[\chi_i(a)\chi_j(a)]=\delta_{ij}$ for all $i,j\in [q]$ where $\delta$ is Kronecker's delta.
  By tensoring this basis, 
  we obtain the orthonormal basis $\{\chi_\rho\}_{\rho\in[q]^k}$ of the vector space $\{f:[q]^k\to \bbR\}$.
  That is,
  for $\rho\in [q]^k, \beta\in [q]^k$,
  we have $\chi_{\rho}(\beta)=\chi_{\rho_1}(\beta_1)\cdots \chi_{\rho_k}(\beta_k)$.
  For a function $f:[q]^k\to \bbR$,
  we define $\hat{f}(\sigma)=\sum_{\beta\in [q]^k}f(x)\chi_{\sigma}(\beta)$.
  Note that $f(\beta)=E_{\sigma\in [q]^k}[\hat{f}(\sigma)\chi_{\sigma}(\beta)]$.
  Therefore, 
  if we let $f$ again be the function corresponding to $\bimu_P$, 
  we have
  \begin{eqnarray*}
    \Pr_{\beta\sim \bimu_P}[\beta_i=a]
    =
    \sum_{\beta\in [q]^k, \beta_i=a}\E_{\sigma\in [q]^k}\left[\widehat{f}(\sigma)\chi_\sigma(\beta)\right]
    =
    \E_{\sigma\in [q]}\left[\widehat{f_i}(\sigma)\chi_\sigma(a)\right].
  \end{eqnarray*}
  Here, $\widehat{f}_i(s)=\widehat{f}(\sigma)$ where $\sigma_i=s$ and $\sigma_r=1$ for all $r\in [k]\setminus\{i\}$.
  In the second inequality, we used that for every $\sigma$ with $\sigma_r\neq 1$ for some $r\in [k]\setminus\{i\}$,
  the sum over the values of $\chi_\sigma$ vanishes.

  We let $g_i:[q]\to \bbR$ be the function $g_i(a)=\bix_{i,a}$.
  We define a function $f':[q]^k\to \bbR$ as follows.
  \begin{eqnarray*}
    \widehat{f'}(\sigma)=\left\{
    \begin{array}{ll}
      \widehat{g}_i(s) & \mbox{if }\sigma_i=s\mbox{ and }\sigma_r=1\mbox{ for all }r\in [k]\setminus \{i\}, \\  
      \widehat{f}(\sigma) & \mbox{otherwise}.      
    \end{array}
    \right.
  \end{eqnarray*}
  This is well-defined since for any $i\in [k]$, 
  it holds that $\widehat{g}_i(1)=\sum_{a\in [q]}g_i(a)=\sum_{a\in [q]}\bix_{i,a}=1$.
  Therefore, the function $f'$ satisfies $\sum_{\beta\in [q]^k}f'(\beta)=\widehat{f}(1)=1$,
  Then, we can define a distribution $\bimu'_P$ corresponding to $f'$,
  and we have
  \begin{eqnarray*}
    \Pr_{\beta\sim \bimu'_P}[\beta_i=a]=\E_{\sigma\in [q]}\left[\widehat{f'_i}(\sigma)\chi_\sigma(a)\right]=\bix_{v,a}.
  \end{eqnarray*}
  Thus, it looks that the $\bimu'_P$ is the desired distribution.
  However, in general, 
  the function $f'$ might take negative values.
  We will show that these values cannot be too negative and that the function can be made to a proper distribution by smoothing.
  
  Let $K$ be an upper bound on the values of the functions $\chi_1,\ldots,\chi_q$.
  From the orthonormality of the functions, 
  it follows that $K\leq \sqrt{q}$.
  Let $f_i(a)=\Pr_{\beta\sim \bimu_P}[\beta_i=a]$.
  Since the LP solution $(\bix,\bimu)$ is $\epsilon$-infeasible,
  we have
  \begin{eqnarray*}
    \left|\widehat{g}_i(s)-\widehat{f}_i(s)\right|
    =
    \left|\sum_{a\in [q]}g_i(a)\chi_s(a)-\sum_{a\in [q]}f_i(a)\chi_s(a)\right|
    \leq
    Kq\epsilon.
  \end{eqnarray*}
  Therefore, $|\widehat{f'}(\sigma)-\widehat{f}(\sigma)|\leq Kq\epsilon$ for all $\sigma\in [q]^k$.
  Recall that $|\widehat{f'}(\sigma)-\widehat{f}(\sigma)|=0$ for $\sigma\in [q]^k$ if there are $i\neq j$ such that $\sigma_i\neq 1, \sigma_j\neq 1$.
  Thus,
  \begin{eqnarray}
    |f'(\beta)-f(\beta)| = \left|\E_{\sigma\in [q]^k}\left[\widehat{f}'(\sigma)\chi_\sigma(\beta) - \widehat{f}(\sigma)\chi_\sigma(\beta)\right]\right|
    \leq
    \delta/q^k,
    \label{eq:f-f'-close}
  \end{eqnarray}
  where $\delta=K^2kq^2\epsilon$.
  Hence, if we let $h=(1-\delta)f'+\delta U$,
  where $U:[q]^k\to\bbR$ is the uniform distribution $U \equiv 1/q^k$,
  then
  \begin{eqnarray*}
    h(x)=(1-\delta)f'(x)+\delta/q^k\geq(1-\delta)f(x)\geq 0.
  \end{eqnarray*}
  It follows that $h$ corresponds to another distribution $\bimu'_P$ over assignments $[q]^k$.
  Furthermore, it holds 
  \begin{eqnarray*}
    \Pr_{\beta\sim \bimu'_P}[\beta_i=a]=(1-\delta)\bix_{i,a}+\frac{\delta}{q}.
  \end{eqnarray*}
  Finally, let us estimate the statistical distance between the distributions $\bimu_P$ and $\bimu'_P$.
  \begin{eqnarray*}
    ||f-h||_1
    =
    ||(1-\delta)(f-f')+\delta(f-U)||_1
    \leq
    ||f-f'||_1+\delta
    \leq
    2\delta.
  \end{eqnarray*}
  The first inequality is from the triangle inequality and the second inequality is from~\eqref{eq:f-f'-close}.
\end{proof}

\begin{proof}[Proof of Lemma~\ref{lmm:robust}]
  Let us consider an $\epsilon$-infeasible LP solution $(\bix,\bimu)$ for a $\Lambda$-CSP instance $\calI$ of value $c\biw_{\calI}$.
  First, we construct vector $\bix'$ as in Lemma~\ref{lmm:surgery}.
  These variables together with the original local distributions $\bimu$ form an $3\epsilon$-infeasible LP solution for $\calI$.
  Next, we construct local distributions $\bimu'$ as in Lemma~\ref{lmm:smoothing}.
  Define new variables
  \begin{eqnarray*}
    \bix''_{i,a}=(1-\delta)\bix'_{i,a}+\delta/q.
  \end{eqnarray*}
  It follows that $(\bix'',\bimu')$ is a feasible LP solution for $\calI$.
  The LP value of this solution is
  \begin{eqnarray*}
    \sum_{P\in \calP}\biw_P\E_{\beta\sim \bimu'_P}[P(\beta)]
    &=&
    c\biw_{\calI}-\sum_{P\in \calP}\biw_P\sum_{\beta\in [q]^{V(P)}}P(\beta)\left(\bimu_{P,\beta}-\bimu'_{P,\beta}\right)\\
    &\geq&
    c\biw_{\calI}-\sum_{P\in \calP}\biw_P||\bimu_P-\bimu'_P||_1\\
    &\geq&
    c\biw_{\calI}-\epsilon \cdot \poly(kq) \biw_{\calI}.
  \end{eqnarray*}
  We used $|P(x)|\leq 1$ for the first inequality, and the second inequality follows from Lemma~\ref{lmm:smoothing}.
\end{proof}

\section{Proof of Lemma~\ref{lmm:packing-lp}}\label{sec:packing-lp}
In this section, we give a proof of Lemma~\ref{lmm:packing-lp}.
We consider a more restricted form of a packing LP:
\begin{eqnarray}
  \begin{array}{ll}
    \max & \mathbf{1}^T\biz \\ 
    \mbox{s.t.} & A^T\biz \leq \bic \\
    & \biz \geq 0,
  \end{array}
  \label{lp:packing-restricted}
\end{eqnarray}
where $A\in \bbR_+^{m \times n}$ is a non-negative matrix such that $a_{ji}=0$ or $a_{ji}\geq 1$ for any $j\in [m], i\in [n]$,
and $\bic\in \bbR_+^{n}$ is a non-negative vector.

Define
\begin{eqnarray*}
  c_{\max} = \max_{i}c_i, \quad
  \Gamma_p = \max_i \frac{c_{\max}}{c_i}\sum_{j=1}^m a_{ji}, \quad
  \Gamma_d = \max_j \sum_{i=1}^na_{ji}.
\end{eqnarray*}
Then, there is a distributed algorithm that solves this packing LP.
\begin{lemma}[\cite{KMW06}]\label{lmm:kmw06}
  For sufficiently small $\epsilon>0$,
  there exists a deterministic distributed algorithm that computes a feasible $(1-\epsilon,0)$-approximate solution to LP~\eqref{lp:packing-restricted} in $O(\log \Gamma_p\log \Gamma_d/\epsilon^4)$ rounds.
  \qed
\end{lemma}

In order to apply Lemma~\ref{lmm:kmw06} to LP~\eqref{lp:blp-relaxed-again},
we transform it to the form LP~\eqref{lp:packing-restricted}.
Note that, in the objective function,
the coefficient of $\bimu_{P,\beta}$ is $\biw_PP(\beta)+C$ and the coefficients of $\bix_{v,a},\overline{\bix}_{v,a},\overline{\bimu}_{P,\beta}$ are $C$.
Thus, by replacing $\bimu_{P,\beta}$ with $\bimu_{P,\beta}/(\biw_PP(\beta)+C)$ and replacing $\bix_{v,a},\overline{\bix}_{v,a},\overline{\bimu}_{P,\beta}$ with $\bix_{v,a}/C,\overline{\bix}_{v,a}/C,\overline{\bimu}_{P,\beta}/C$, respectively, we obtain the following LP.
\begin{eqnarray*}
  \begin{array}{lll}
    \max & \mathbf{1}^T(\bix+\overline{\bix}+\bimu+\overline{\bimu}) \\
    \mbox{s.t.} & \sum\limits_{a\in[q]} \frac{\bix_{v,a}}{C} \leq q-1 + \epsilon & \forall v\in V\\    
    & \sum\limits_{a\in[q]}\frac{\overline{\bix}_{v,a}}{C} \leq 1+\epsilon & \forall v\in V\\    
    & \frac{\bix_{v,a}}{C}+\sum\limits_{\beta \in [q]^{V(P)}, \beta_v=a}\frac{\bimu_{P,\beta}}{\biw_PP(\beta)+C} \leq 1+\epsilon & \forall P\in \calP, v\in V(P), a\in [q]\\
    & \frac{\overline{\bix}_{v,a}}{C}+\sum\limits_{\beta \in [q]^{V(P)}, \beta_v=a}\frac{\overline{\bimu}_{P,\beta}}{C} \leq q^{V(P)-1}+\epsilon & \forall P\in \calP, v \in V(P), a\in [q]\\
    & \frac{\bix_{v,a}}{C}+\frac{\overline{\bix}_{v,a}}{C} \leq 1, \quad \bix_{v,a}\geq 0, \quad \overline{\bix}_{v,a}\geq 0 & \forall v \in V\\
    & \frac{\bimu_{P,\beta}}{\biw_PP(\beta)+C}+\frac{\overline{\bimu}_{P,\beta}}{C} \leq 1, \quad \bimu_{P,\beta}\geq 0, \quad \overline{\bimu}_{P,\beta}\geq 0 & \forall P\in \calP, \beta\in [q]^{V(P)}.\\
  \end{array}
\end{eqnarray*}

We multiply each constraint in order to make every coefficient in the LHS at least $1$.
Then, we have the following LP.
\begin{eqnarray}
  \begin{array}{lll}
    \max & \mathbf{1}^T(\bix+\overline{\bix}+\bimu+\overline{\bimu}) \\
    \mbox{s.t.} & \sum\limits_{a\in[q]}\bix_{v,a} \leq C(q-1 + \epsilon) & \forall v \in V\\    
    & \sum\limits_{a\in[q]}\overline{\bix}_{v,a} \leq C(1+\epsilon) & \forall v\in V\\    
    & \frac{w+C}{C}\bix_{v,a}+\sum\limits_{\beta \in [q]^{V(P)}, \beta_v=a}\frac{w+C}{\biw_PP(\beta)+1}\bimu_{P,\beta} \leq (1+\epsilon)(w+C) & \forall P\in \calP, v\in V(P), a\in [q]\\
    & \overline{\bix}_{v,a}+\sum\limits_{\beta \in [q]^{V(P)}, \beta_v=a}\overline{\bimu}_{P,\beta} \leq C(q^{V(P)-1}+\epsilon) & \forall P\in \calP, v\in V(P), a\in [q]\\
    & \bix_{v,a}+\overline{\bix}_{v,a} \leq C, \quad \bix_{v,a}\geq 0, \quad \overline{\bix}_{v,a}\geq 0 & \forall v \in V\\
    & \frac{w+C}{\biw_PP(\beta)+1}\bimu_{P,\beta}+\frac{w+C}{C}\overline{\bimu}_{P,\beta} \leq w+C, \quad \bimu_{P,\beta}\geq 0, \quad \overline{\bimu}_{P,\beta}\geq 0 & \forall P\in \calP, \beta\in [q]^{V(P)}.\\
  \end{array}
  \label{lp:blp-super-relaxed}
\end{eqnarray}

\begin{proof}[Proof of Lemma~\ref{lmm:packing-lp}]
  Note that LP~\eqref{lp:blp-super-relaxed} is of the form LP~\eqref{lp:packing-restricted}.
  After a calculation, we have
  \begin{eqnarray*}
    c_{\max}  = O(C(w+q^s)), 
    \quad
    \Gamma_p  = O((s+t)(w+C)\cdot C(w+q^s)),  
    \quad
    \Gamma_d  = O((w+C)q^s).
  \end{eqnarray*}
  We define the \textit{degree} of a variable in an LP as the number of inequalities where the variable appears.
  Let $\Delta_p$ and $\Delta_d$ be the maximum degree of primal variables and dual variables, respectively.
  Here, we treat LP~\eqref{lp:blp-super-relaxed} as a dual formulation.
  We have
  \begin{eqnarray*}
    \Delta_p  = O(q^s),  
    \quad
    \Delta_d  = O(s+t).
  \end{eqnarray*}

  Applying the algorithm given in Lemma~\ref{lmm:kmw06} to LP~\eqref{lp:blp-super-relaxed},
  we obtain a distributed algorithm that calculates $(1-\epsilon,0)$-approximate solution.
  The number of rounds is $O(\log \Gamma_p\log \Gamma_d/\epsilon^4)$.
  Note that,
  given a variable,
  we can simulate the computation of the distributed algorithm involved by the variable with $(\Delta_p\Delta_d)^r$ queries,
  where $r$ is the number of rounds.
  Thus, the query complexity becomes
  \begin{eqnarray*}
    (\Delta_p\Delta_d)^{O(\log \Gamma_p\log \Gamma_d/\epsilon^4)}
    =
    \exp(\poly(qstw/\epsilon)).
  \end{eqnarray*}
\end{proof}

\section{Proofs from Section~\ref{sec:round}}\label{apx:round-appendix}
\subsection{Proof of Lemma~\ref{lmm:discretize}}
\begin{proof}
  Since we move each $\bix_{v,a}$ by at most $\epsilon$,
  each constraint $\sum_{a\in [q]} \bix^\epsilon_{v,a}=1$ can be at most $(q+1)\epsilon$-infeasible.
  Also, each constraint $\sum_{\beta\in [q]^{V(P)},\beta_v=a} \bimu_{P,\beta}=\bix^\epsilon_{v,a}$ can be at most $2\epsilon$-infeasible.
\end{proof}

\subsection{Proof of Lemma~\ref{lmm:compression}}
\begin{proof}
  Since the size of the range of $\phi_{\bix}$ is $(1/\epsilon)^{O(q)}$, the second claim is obvious.

  Suppose that $(\bix,\bimu)$ has an LP value $c\biw_{\calI}$.
  From the fact that $(\bix,\bimu)$ is a $(1-\epsilon,\epsilon)$-approximate solution, 
  we have $c\biw_{\calI} \geq (1-\epsilon)\lp(\calI)-\epsilon n$.
  Also, by Lemma~\ref{lmm:discretize}, 
  $(\bix^\epsilon,\bimu)$ is a $(q+1)\epsilon$-infeasible LP solution.
  Since only $\bimu$ affects the value of the objective function, 
  the LP value of $(\bix^\epsilon,\bimu)$ equals $c\biw_{\calI}$.
  A key observation is that $(\bix^\epsilon,\bimu)$ is also an LP solution for the folded instance $\calI/\phi_{\bix}$.
  Thus, we see that $\calI/\phi_{\bix}$ has a $(q+1)\epsilon$-infeasible solution of value at least $c\biw_{\calI}$.
  From Lemma~\ref{lmm:robust}, we have
  \begin{eqnarray*}
    \lp(\calI/\phi_{\bix})
    &\geq&
    (c - (q+1)\epsilon \cdot \poly(qs)) \biw_{\calI}  \\
    &\geq&
    (1-\epsilon)\lp(\calI)-\epsilon n -  \epsilon\cdot \poly(qs) \biw_{\calI} \\
    &\geq &
    \lp(\calI) - \epsilon \cdot \poly(qstw)n.
  \end{eqnarray*}
  In the last inequality,
  we use the fact that $\lp(\calI)\leq \biw_{\calI} \leq twn$.
\end{proof}

\section{Proofs from Section~\ref{sec:lower}}\label{apx:lower-appendix}

\subsection{Proof of Lemma~\ref{lmm:less-than-opt}}\label{sec:less-than-opt}
Let $\calJ_P$ be an instance generated by $\calD_{N,T}^{\opt}(P)$.
Let $P_i(1\leq i\leq 2)$ be a constraint on a variable sequence $\{u_1^i,\ldots,u_k^i\}$ in $\calJ_P$.
Note that the arities of $P_i$ are the same since they both are copies of $P$.
For each $j\in [k]$,
we choose $v_j^1 \in \{u_j^1,u_j^2\}$ arbitrarily and $v_j^2$ be the remaining one, i.e., $\{u_j^1,u_j^2\} \setminus \{v_j^1\}$.
Then, we define a constraint $Q_i(1\leq i\leq 2)$ on the variable sequence $\{v_1^i,\ldots,v_k^i\}$.
We create another instance $\calJ'_P$ from $\calJ_P$ by replacing $\{P_1,P_2\}$ by $\{Q_1,Q_2\}$.
We call this method \textit{switching}.
The following concentration bound is obtained by a simple application of Theorem~2.19 in~\cite{Wor99}.
\begin{lemma}\label{lmm:chernoff-reg}
  If $\bfX$ is a random variable defined on $\calD_{N,T}^{\opt}(P)$ such that $|\bfX(\calJ_P)-\bfX(\calJ'_P)|\leq c$ holds where $\calJ_P$ and $\calJ'_P$ are instances of $\calD_{N,T}^{\opt}(P)$ that only differ by a switching,
  then
  \begin{eqnarray*}
    \Pr_{\calJ_P\sim \calD_{N,T}^{\opt}(P)}\left[ \left|\bfX(\calJ_P)-\E[\bfX(\calJ_P)] \right| \geq t\right] \leq 2\exp\left(-\frac{t^2}{TN c^2}\right)
  \end{eqnarray*}
  for all $t>0$.
  \qed
\end{lemma}

\begin{proof}[Proof of Lemma~\ref{lmm:less-than-opt}]
  Let $\alpha \in [q]^{V\times [N]}$ be an assignment to $\calJ$.
  For $v\in V$ and $a\in [q]$, we define $\bix_{v,a}=\#\{i \in [N] \mid \alpha_{(v,i)}=a\}/N$.
  Also, for $P\in \calP$ and $\beta\in [q]^{V(P)}$, we define $\bimu_{P,\beta}=\prod_{v\in P}\bix_{v,\beta_v}$.
  Note that $\bix_v$ (resp., $\bimu_P$) gives a probability distribution over assignments to the variable $v$ (resp., the variable set $V(P)$).

  Let $\calJ_P$ be the sub-instance of $\calJ$ generated by $\calD_{N,T}^{\opt}(P)$ for $P\in \calP$.
  The expectation (over $\calD_{N,T}^{\opt}(P)$) of the value gained by a constraint $P$ in $\calJ_P$ is $\E_{\beta_P\sim \bimu_P}\left[P(\beta_P)\right]$.
  Thus, it holds that
  \begin{eqnarray*}
    &&\E_{\calJ \sim \calD_{N,T}^{\opt}}[\val(\calJ,\alpha)] 
    =
    \sum_{P\in \calP}\E_{\calJ_P \sim \calD_{N,T}^{\opt}(P)}[\val(\calJ_P,\alpha_{|V(P)})]
    =
    TN\sum_{P\in \calP}\biw_P  \E_{\beta_P \sim \bimu_P}[P(\beta_P)] \\
    &=&
    TN\sum_{P\in \calP}\biw_P \E_{\beta \sim \bimu}[P(\beta_{|V(P)})]
    =
    TN\E_{\beta \sim \bimu}[\sum_{P\in \calP}\biw_PP(\beta_{|V(P)})]
    =
    TN\E_{\beta \sim \bimu}[\val(\calI,\beta)]
  \end{eqnarray*}
  Thus, it follows that
  \begin{eqnarray}
    \E_{\calJ \sim \calD_{N,T}^{\opt}}[\val(\calJ,\alpha)]  \leq TN\opt(\calI). \label{eq:less-than-opt-1}
  \end{eqnarray}

  Note that, for instances $\calJ_P$ and $\calJ'_P$ generated by $\calD_{N,T}^{\opt}(P)$ such that they differ by a switching,
  $\val(\calJ_P,\alpha_{|V(P)})$ and $\val(\calJ_P',\alpha_{|V(P)})$ can differ by at most $2\biw_P\leq 2w$.
  Then, from Lemma~\ref{lmm:chernoff-reg}, 
  \begin{eqnarray*}
    \Pr\left[\left|\val(\calJ_P,\alpha_{|V(P)})-\E[\val(\calJ_P,\alpha_{|V(P)})]\right|\geq t \right]
    \leq
    2\exp\left(-\frac{t^2}{4TN w^2}\right).
  \end{eqnarray*}
  Then,
  \begin{eqnarray*}
    \Pr \left[|\val(\calJ,\alpha)-\E[\val(\calJ,\alpha)]|\geq t|\calP| \right]
    &\leq &
    \Pr \left[\exists P\in \calP,\; |\val(\calJ_P,\alpha_{|V(P)})-\E[\val(\calJ_P,\alpha_{|V(P)})]|\geq t \right] \\
    &\leq&
    2|\calP|\exp\left(-\frac{t^2}{4TN w^2}\right).
  \end{eqnarray*}
  The last inequality is from the union bound.

  We choose $t = \epsilon TN$ so that $t|\calP| = \epsilon |\calP|TN \leq \epsilon \biw_{\calJ}$.
  We have
  \begin{eqnarray}
    \Pr\Bigl[|\val(\calJ,\alpha)-\E[\val(\calJ,\alpha)]|\geq \epsilon \biw_{\calJ}  \Bigr]
    \leq
    2|\calP|\exp\left(-\frac{\epsilon^2 TN}{4w^2}\right). \label{eq:less-than-opt-2}
  \end{eqnarray}
  
  We combine \eqref{eq:less-than-opt-1} and \eqref{eq:less-than-opt-2} with the union bound over all $q^{|V|N}$ assignments.
  It holds that
  \begin{eqnarray*}
    \Pr\Bigl[\exists \alpha, \val(\calJ,\alpha) \geq TN\opt(\calI) + \epsilon \biw_{\calJ} \Bigr] 
    &=&
    \Pr\Bigl[\exists \alpha, \olval(\calJ,\alpha) \geq \olopt(\calI) + \epsilon \Bigr]\\
    &\leq&
    2|\calP|\exp\left(-\frac{\epsilon^2 TN}{4w^2}\right) q^{|V|N}.
  \end{eqnarray*}
  by choosing $T=\Theta(w^2\log q/\epsilon^2)$, we have the desired result.
  Note that $|V|$ and $|\calP|$ can be seen as constants when $N$ is sufficiently large.
\end{proof}

\subsection{Proof of Lemma~\ref{lmm:more-than-lp}}\label{sec:more-than-lp}
\begin{proof}
  Let $\alpha \in [q]^{V\times [N]}$ be the natural assignment to variables in $\calJ$.
  That is, $\alpha(v,i)=a$ when the variable $(v,i)$ is assigned to the value $a$ in the construction of $\calD_{N,T}^{\lp}$.
  Then,
  \begin{eqnarray*}
    \opt(\calJ) 
    \geq
    \val(\calJ,\alpha)
    =
    \sum_{P\in \calP}\val(\calJ_P,\alpha_{|V(P)}) 
    =
    TN\sum_{P\in \calP} \biw_P \E_{\beta_P \sim \bimu^*_P}[P(\beta_P)]
    =
    TN\lp(\calI).
  \end{eqnarray*}
\end{proof}

\subsection{Proof of Lemma~\ref{lmm:distinguish}}\label{sec:distinguish}
For notational simplicity, we omit subscripts $N$ and $T$ in this section.
We define some notions.
At each step of an algorithm,
a variable $v$ is called \textit{seen} if $v$ is appeared in queries to the oracle or answers by the oracle so far. 
Also, an index $i$ of a variable $v$ is called \textit{seen} if the $i$-th constraint of $v$ is already returned by the oracle.

Here, we only show a lower bound for a (randomized) algorithm whose behavior is slightly restricted.
That is, when an algorithms asks for a constraint incident to an unseen variable, 
we assume that the algorithm chooses the variable uniformly at random from the set of unseen variables.
We can get rid of this assumption using the technique presented in Section~4 of~\cite{GT03}.
Details are deferred to the full version of the paper.
In what follows, we regard that the oracle accepts two types of queries.
The first one is same as the original, i.e., when we specify a variable $v$ and an index $i$, the oracle returns the $i$-th constraint of $v$.
The second one simply returns a random variable from the set of unseen variables without receiving any argument.
When an algorithm asks for a constraint incident to an unseen variable,
it uses the second type of queries to get a variable first, 
and then it uses the first type of queries to get a constraint incident to the variable.

Now, we prove Lemma~\ref{lmm:distinguish}.
Recall that, from Yao's minimax principle, it suffices to consider deterministic algorithms.
We basically follow the approach presented in Section~7 of~\cite{GR08}.
Let $\calA$ be a deterministic algorithm.
We introduce a randomized process $\calP^{\opt}$ (resp.,~$\calP^{\lp})$, 
which interacts with $\calA$ so that $\calP^{\opt}$ (resp.,~$\calP^{\lp}$) answers queries of $\calA$ to the oracle while constructing a random instance from $\calD^{\opt}$ (resp.,~$\calD^{\lp})$.
The final distribution of instances generated by $\calP^{\opt}$ (resp.,~$\calP^{\lp}$) coincides with $\calD^{\opt}$ (resp.,~$\calD^{\lp}$) no matter how $\calA$ makes queries.
The interaction between $\calA$ and $\calP^{\opt}$ (resp.,$\calP^{\lp}$) precisely simulates the interaction between $\calA$ and $\calO_{\calJ}$ where $\calJ$ is an instance generated by the distribution $\calD^{\opt}$ (resp.,~$\calD^{\lp}$).
The process $\calP^\star$, which corresponds to the distribution $\calD^\star$,
is simply a process that chooses $\calP^{\opt}$ or $\calP^{\lp}$ randomly and behave as the chosen process.

A \textit{transcript} is the part of an instance that $\calA$ has seen through the interaction with a randomized process.
Note that, the transcript contains the information about labels of vertices and indices of constraints.
Let $\calK_{\tau}^{\opt}$ (resp.,~$\calK_{\tau}^{\lp})$ be the distribution of transcripts after $\tau$-step interaction between $\calA$ and $\calP^{\opt}$ (resp.,~$\calP^{\lp}$) (Here, $\calK$ stands for \textit{knowledge}).
The statistical distance between $\calK_{\tau}^{\opt}$ and $\calK_{\tau}^{\lp}$ is defined as follows.
\begin{eqnarray*}
  \dtv(\calK_{\tau}^{\opt},\calK_{\tau}^{\lp}) = \frac{1}{2}\sum_{K}\left| \Pr_{K' \sim \calK_{\tau}^{\opt}}[K' = K] - \Pr_{K' \sim \calK_{\tau}^{\lp}}[K' = K] \right|
\end{eqnarray*}
From the argument given in Section~7 of~\cite{GR08},
by showing that $\dtv(\calK_{\tau}^{\opt},\calK_{\tau}^{\lp}) = o(1)$ when $\tau=o(\sqrt{N})$,
we have the desired result.

We can safely assume that $\calA$ never asks for the same constraint twice or more.
Also, we assume that,
if $\calP^{\star}$ returns a constraint containing a variable in the transcript,
$\calA$ can correctly guess the process ($\calP^{\opt}$ or $\calP^{\lp}$) with which $\calA$ is interacting.
In other words,
we are assuming that, 
when $\calP^{\opt}$ (resp.,~$\calP^{\lp}$) returns a constraint containing a variable in the transcript,
it also returns a certificate stating that the current process is $\calP^{\opt}$ (resp.,~$\calP^{\lp}$).
This only improves the ability of $\calA$ and makes the lower bound smaller.

Now, we define the randomized process $\calP^{\opt}$.
We omit the definition of $\calP^{\lp}$ as it is very similar to the construction of $\calP^{\opt}$.
The process $\calP^{\opt}$ has two stages.
The first stage proceeds as long as $\calA$ perform queries.
In this stage, 
$\calP^{\opt}$ chooses an answer for each query.
In the second stage, 
the process completes the transcript into an instance $\calJ$.

We identify $[n]$ (resp., $[nN]$) with the set of variables of $\calI$ (resp., an instance generated by $\calP^{\opt}$).
Recall that,
in an instance generated by $\calD^{\opt}$,
the variable set $[nN]$ can be separated into $n$ sets,
each of which corresponds to a variable $i\in [n]$.
The process $\calP^{\opt}$ incrementally constructs this correspondence.
A (partial) correspondences is represented by a map $\rho \colon [nN] \to [n] \cup \{\bot\}$.
For a variable $i \in [n]$,
let $V_i = \{v \in [nN] \mid \rho(v) = i\}$ and $N_i = |V_i|$.
Also, for each vertex $v \in [nN]$ and an index $i \in [d]$,
let $D_{i}(v) = \{j \in \{(T-1)i+1, \ldots, Ti \} \mid \text{$j$-th constraint of $v$ is seen}\}$ and $d_i(v) = |D_i(v)|$.

In the first stage,
given a query by an algorithm $\calA$, 
$\calP^{\opt}$ chooses an answer for it as follows.
\begin{itemize}
\setlength{\itemsep}{0pt}
\item When the query asks for a random unseen variable: 
  we choose a random unseen variable $v \in [nN]$, 
  and set $\rho(v) = i$ with probability $\frac{N-N_i}{\sum_{j\in [n]} (N-N_j)}$.
  Then, we return $v$ to $\calA$.
\item When the query asks for the $p$-th constraint of $v$:
  Note that $\rho(v) \neq \bot$ from the assumption that, when $\calA$ asks for a constraint incident to an unseen variable, it asks for a random unseen variable beforehand.
  Let $q$ be such that $(T-1)q+1 \leq p \leq Tq$, and $P$ be the $q$-th constraint of $\rho(v)$ in $\calI$,
  which is applied to a sequence of variables $\{i_1,\ldots,i_k\}$ in $\calI$ for which $i_\ell = \rho(v)$ for some $\ell \in [k]$.
  Also, let $q_j$ be such that $P$ is the $q_j$-th constraint of the variable $i_j$ in $\calI$.
  Note that $q_\ell = q$.
  
  Then, we choose a set of variables $\{v_j\}_{j \in [k] \setminus \{\ell\}}$ as follows.
  For each variable $u$ with $\rho(u)=i_j$,
  we choose $u$ as $v_j$ with probability $\frac{T-d_{q_j}(u)}{\sum_{w\in V_{i_j}}(T-d_{q_j}(w)) + (N-{N_{i_j}})T}$.
  If otherwise, we choose a random unused variable $u$ as $v_j$ ans set $\rho(u)=i_j$.

  Let $P'$ be a constraint applied to a sequence $\{v_1,\ldots,v_k\}$ of weight $\biw_{P}$.
  Finally, we determine indices for each variable $v_j (j \neq \ell)$.
  We choose a random index $p_j$ from unused indices in $\{(T-1)q_j+1,\ldots,Tq_j+1\}$,
  and set $P'$ be as the $p_j$-th constraint of $v_j$.
  Then, we return $P'$ as the answer for the query.
\end{itemize}
In the second stage of $\calP^{\opt}$, 
the process uniformly selects an instance $\calJ$ among all those who are consistent with the final transcript.

\begin{lemma}\label{lmm:equivalent}
  For every algorithm $\calA$, 
  the randomized process $\calP^{\opt}$ (resp., $\calP^{\lp}$) when interacting with $\calA$, 
  uniformly generates an instance $\calJ$ in $\calD^{\opt}$ (resp., $\calD^{\lp}$).
  \qed
\end{lemma}
\begin{proof}
  The lemma easily follows by induction on the query complexity of $\calA$.
  The base case is clear since if no query is made, 
  then the distribution on instances generated by $\calP^{\opt}$ (or, $\calP^{\lp}$) is clearly uniform.
  The induction step follows directly from the definition of the process.
  In particular, 
  the distribution on instances resulting from the process switching to the second stage after it answers the query is exactly the same as the distribution resulting from the process performing the second stage without answering the query.
\end{proof}

\begin{proof}[Proof of Lemma~\ref{lmm:distinguish}]
  Let $\calA$ be a deterministic algorithm.
  It is convenient to think that labels of variables are determined on the fly.
  That is, $\calP^{\star}$ decides labels of variables from $[nN]$ at the time when the variable appears for the first time in the interaction between an algorithm and $\calP$.
  The distribution never change by this modification.
  Also, 
  we can think that the sequence of labels is determined beforehand,
  and for each time when a new variable appears,
  a new label for the variable is taken from the front of the sequence.
  Let $\calP^{\star}_{\ell}$ be the process obtained from $\calP^{\star}$ by fixing the sequence to $\ell$.
  It is clear that $\calP^{\star}$ coincides with the process that takes $\ell$ uniformly at random and acts as $\calP^{\star}_{\ell}$.
  Let $\calP^{\opt}_{\ell}$ (resp.,~$\calP^{\lp}_{\ell}$) be the process obtained from $\calP^{\opt}$ (resp.,~$\calP^{\lp}$) by fixing the sequence to $\ell$.
  Then, it suffices to bound the statistical distance between the distribution of transcripts when $\calA$ interacts with $\calP^{\opt}_{\ell}$ and the one when $\calA$ interacts with $\calP^{\lp}_{\ell}$ for any sequence $\ell$.

  A deterministic algorithm $\calA$ with query complexity $\tau$ can be expressed as a decision tree of depth at most $\tau$.
  Here, each node in the decision tree corresponds to a query to the oracle,
  and each branch from the node corresponds to the answer by the oracle.
  Recall that, from the rule of indices,
  if we fix an index,
  the process always returns the same predicate (though the set of vertices to which the predicate is applied should differ).
  Also, 
  since we have fixed the sequence of labels $\ell$,
  at each node in the decision tree,  
  there is just one branch corresponding to the case that $\calA$ finds a constraint such that any variable in the constraint (except the queried variable) is not in the transcript.
  Ignoring branches for which $\calA$ outputs an answer,
  the decision tree has the property that the number of children of each node is at most one.
  Thus, $\calA$ is essentially a non-adaptive algorithm.
  Without loss of generality,
  we assume that $\calA$ outputs that the current instance is generated by $\calD^{\opt}$ after $\tau$ steps.
  
  Suppose that the current process is $\calP^{\opt}_{\ell}$ and $\calA$ is asking for a constraint incident to some variable in the $i$-th query.
  Note that $\calA$ has seen at most $is$ variables.
  Then, from the construction of $\calP^{\opt}$,
  the probability that $\calP^{\opt}_{\ell}$ returns a variable in the transcript is at most $\frac{isT}{(N-is)T}\cdot s = \frac{is^2}{N-is}$.
  Using the same argument,
  we can show that,
  in the $i$-th query,
  the probability that $\calP^{\lp}_{\ell}$ returns a variable in the transcript is at most $\frac{is^2}{\mu N-is}$ where $\mu$ is the minimum of $\{\bimu_{P,\beta}\}_{P \in \calP, \beta \in [q]^{V(P)}}$ except $0$.
  
  Thus, from the union bound,
  after $\tau$ steps,
  the probability that $\calP^{\star}_{\ell}$ returns a variable in the transcript is at most 
  \begin{eqnarray*}
    \sum_{i=1}^{\tau} \frac{is^2}{\mu N-is} \leq \frac{\tau^2s^2}{\mu N-\tau s}.
  \end{eqnarray*}
  Then, the probability that $\calA$ outputs the correct answer is at most $\frac{\tau^2s^2}{\mu N-\tau s} + \frac{1}{2}$.
  To make this probability at least $3/5$, we have to choose $\tau = \Omega(\sqrt{N})$.
  Note that $\mu$ is a positive constant independent of $N$.
\end{proof}

\section{Proof of Theorem~\ref{thr:prop}}\label{apx:prop}
\begin{proof}
  We show the first part of the theorem.
  Let $\Lambda$ be a CSP such that $S_{\Lambda}(1) = 1-\gamma$ for some $\gamma>0$.
  Suppose that there exists a testing algorithm for the CSP~$\Lambda$ with $o(\sqrt{n})$ queries.
  Note that a $\frac{\gamma}{3}$-far instance $\calI$ satisfies that $\opt(\calI) \leq \biw_{\calI}-\frac{\gamma twn}{3} \leq (1 - \frac{\gamma}{3})\biw_{\calI}$.
  Thus, using the testing algorithm, given an instance $\calI$, 
  with probability at least $2/3$,
  we can distinguish the case $\opt(\calI) = \biw_{\calI}$ from the case $\opt(\calI) \leq (1-\frac{\gamma}{3})\biw_{\calI}$.
  However, instantiating Theorem~\ref{thr:lower} with $\epsilon = \gamma/3$,
  the theorem asserts that
  any algorithm that, given an instance $\calI$, 
  with probability at least $2/3$,
  distinguishes the case $\opt(\calI) = \biw_{\calI}$ from the case $\opt(\calI) \leq (S_{\Lambda}(1)+\frac{\gamma}{3})\biw_{\calI} = (1-\frac{2\gamma}{3})\biw_{\calI}$ requires $\Omega(\sqrt{n})$ queries.
  This is a contradiction.
  
  We show the second part of the theorem.
  Let $\Lambda$ be a CSP such that $S_{\Lambda}(1) = 1$.
  Since $S_{\Lambda}(c)$ is continuous at $c=1$,
  for any $\epsilon>0$, there exists $\delta$ such that 
  $S_{\Lambda}(1-\delta) > 1-\epsilon/2$.
  Consider the algorithm obtained by instantiating Theorem~\ref{thr:upper} replacing $\epsilon$ with $\min(\epsilon/2,\delta)$.
  Suppose that $\calI$ is a satisfiable instance.
  Then, we obtain a value $x \geq S_{\Lambda}(1-\delta)\biw_{\calI}- \epsilon n/2 > (1-\epsilon/2)\biw_{\calI} - \epsilon n/2 $.
  Suppose that $\calI$ is an instance $\epsilon$-far from satisfiability.
  Then, we obtain a value $x \leq \opt(\calI) \leq \biw_{\calI} - \epsilon twn \leq (1-\epsilon/2)\biw_{\calI} - \epsilon twn/2$.
  Thus, we can test the satisfiability of the CSP~$\Lambda$ in constant time.
\end{proof}

\section*{Acknowledgements}
The author is grateful to Hiro Ito and Suguru Tamaki for valuable comments on an earlier draft of this paper.

\end{document}